\documentclass[runningheads]{llncs}
\usepackage{graphicx}
\usepackage{xcolor}
\usepackage{amsmath,amssymb,amsfonts}
\usepackage{tcolorbox}
\usepackage{xspace}
\usepackage{array}
\usepackage{soul} % texte barré avec \st{}
\usepackage[normalem]{ulem} % texte barré avec \sout{}
\usepackage{tabu}

\usepackage{tikz} 
\usepackage{pgfplots}
\newcolumntype{L}[1]{>{\raggedright\let\newline\\\arraybackslash\hspace{0pt}}m{#1}}
\newcolumntype{C}[1]{>{\centering\let\newline\\\arraybackslash\hspace{0pt}}m{#1}}
\newcolumntype{R}[1]{>{\raggedleft\let\newline\\\arraybackslash\hspace{0pt}}m{#1}}

\usepackage{bm}
\usepackage[ruled,vlined,linesnumbered]{algorithm2e}
\usepackage{booktabs}
\usepackage{dsfont} %for \mathds{1}_{n_2}
\usepackage{multirow}
\usepackage{graphicx}
\usepackage{placeins}
\usepackage{pdf14}

\definecolor{vert}{rgb}{0,0.546875,0}
\definecolor{rouge}{rgb}{0.775,0.2,0.2}

\DeclareMathOperator{\rank}{Rank}
\DeclareMathOperator{\Supp}{Supp}

\newcommand{\ff}[1]{\mathbb{F}_{#1}} % finite field

\newcommand{\lar}{\stackrel{\mathdollar}{\leftarrow}}
\newcommand{\KeyGen}{\normalfont\textsf{KeyGen}}
\newcommand{\Encrypt}{\normalfont\textsf{Encrypt}}
\newcommand{\Decrypt}{\normalfont\textsf{Decrypt}}

\newcommand{\INDCPA}{\ensuremath{\normalfont\textsf{IND\text{-}CPA}}\xspace}

\newcommand{\RSD}{\textsf{RSD}\xspace}
\newcommand{\DRSD}{\textsf{DRSD}\xspace}
\newcommand{\DNHRSD}{\textsf{DNHRSD}\xspace}
\newcommand{\RSL}{\textsf{RSL}\xspace}
\newcommand{\IRSL}{\textsf{IRSL}\xspace}

\newcommand{\param}{\ensuremath{\mathsf{param}}}	
\newcommand{\Setup}{\ensuremath{\mathsf{Setup}}}

\newcommand{\T}{\mathcal{T}}

\newcommand{\sk}{\ensuremath{\mathsf{sk}}\xspace}

\newcommand{\pk}{\ensuremath{\mathsf{pk}}\xspace}

\newcommand{\fold}{\mathsf{Fold}}
\newcommand{\unfold}{\mathsf{Unfold}}

\usepackage{marvosym}
\usepackage{wasysym}
\usepackage{pdfpages}

\usepackage{arydshln}

\newcommand{\norme}[1]{\| #1 \|}

\newcommand{\IRSD}{\textsf{IRSD}\xspace}
\newcommand{\DIRSD}{\textsf{DIRSD}\xspace}
\newcommand{\NHRSD}{\textsf{NHRSD}\xspace}
\newcommand{\NHRSL}{\textsf{NHRSL}\xspace}
\newcommand{\DRSL}{\textsf{DRSL}\xspace}
\newcommand{\DNHRSL}{\textsf{DNHRSL}\xspace}
\newcommand{\NHIRSD}{\textsf{NHIRSD}\xspace}
\newcommand{\NHIRSL}{\textsf{NHIRSL}\xspace}

\newcommand{\DNHIRSL}{\textsf{DNHIRSL}\xspace}

\newcommand{\N}{\mathbb{N}}
\newcommand{\Fq}{\mathbb{F}_q}

\newcommand{\fq}{\Fq}
\newcommand{\Fqm}{\mathbb{F}_{q^m}}

\newcommand{\fqm}{\Fqm}
\newcommand{\C}{{\mathcal{C}}}

\newcommand{\word}[1]{\ensuremath{\boldsymbol{#1}}}
\renewcommand\word[1]{\mathbf{#1}}

\newcommand{\cv}{\word{c}}

\newcommand{\ev}{\word{e}}

\newcommand{\gv}{\word{g}}
\newcommand{\hv}{\word{h}}

\newcommand{\mv}{\word{m}}

\newcommand{\sv}{\word{s}}

\newcommand{\uv}{\word{u}}
\newcommand{\vv}{\word{v}}

\newcommand{\xv}{\word{x}}
\newcommand{\yv}{\word{y}}

\DeclareMathOperator{\Mat}{Mat}

\newcommand{\mat}[1]{\boldsymbol{#1}} % font for matrix
\newcommand{\probab}[2]{\Pr_{\substack{#2}}\left[#1 \right]}

\newcommand{\trsp}[1]{#1^\mathsf{T}} % transpose matrix
\newcommand{\Hm}{\mat{H}}

\newcommand{\Sm}{\mat{S}}
\newcommand{\Gm}{\mat{G}}
\newcommand{\Cm}{\mat{C}}

\newtheorem{ass}{Assumption}
\newtheorem{probe}{Problem}

\newtheorem{coro}{Corollary}

\newcommand{\Idm}{\mat{I}}

\newcommand{\IM}[1]{\mathcal{IM}(#1)}

\newcommand{\schemeone}{Multi-RQC-AG~}
\newcommand{\schemetwo}{Multi-UR-AG~}

\usepackage{environ}

\newcommand{\repeattheorem}[1]{%
  \begingroup
  \renewcommand{\thetheorem}{\ref{#1}}%
  \expandafter\expandafter\expandafter\theorem
  \csname reptheorem@#1\endcsname
  \endtheorem
  \endgroup
}

\NewEnviron{reptheorem}[1]{%
  \global\expandafter\xdef\csname reptheorem@#1\endcsname{%
    \unexpanded\expandafter{\BODY}%
  }%
  \expandafter\theorem\BODY\unskip\label{#1}\endtheorem
}

\usepackage{multirow}
\usepackage[colorlinks]{hyperref}
\hypersetup{
    colorlinks,
    linkcolor={red!80!black},
    citecolor={green!60!black},
    urlcolor={blue!80!black}
}

\begin{document}
\title{RQC revisited and more cryptanalysis 
for Rank-based Cryptography}

\author{}
\institute{}


\author{
Loïc Bidoux\inst{1} 
\and Pierre Briaud \inst{2,3} \and
Maxime Bros\inst{4}
\and
Philippe Gaborit\inst{4}
}
\authorrunning{L. Bidoux et al.}

\institute{
Cryptography Research Center, Technology Innovation Institute, Abu Dhabi, UAE \\
\email{loic.bidoux@tii.ae}\\
\and
Inria, 2 rue Simone Iff, 75012 Paris, France\
	\and
	Sorbonne Universit\'es, UPMC Univ Paris 06 \\
	\email{pierre.briaud@inria.fr}
\and
University of Limoges, CNRS, XLIM, UMR 7252, Limoges, France \\
\email{\{maxime.bros, philippe.gaborit\}@unilim.fr}
}
\maketitle              % typeset the header of the contribution
\begin{abstract}

We propose two main contributions: first, we revisit the encryption scheme
Rank Quasi-Cyclic (RQC) \cite{AABBBBCDGHZ20} by introducing new efficient 
variations, in particular, a new class of codes, 
the Augmented Gabidulin codes; second, we propose new attacks against 
the Rank Support Learning (\RSL), the Non-Homogeneous 
Rank Decoding (\NHRSD), \ and the Non-Homogeneous Rank 
Support Learning (\NHRSL) problems. 
\RSL is primordial for all recent 
rank-based cryptosystems such as Durandal \cite{ABGHZ19} or
LRPC with multiple syndromes \cite{LRPC2022}, 
moreover, \NHRSD and \NHRSL, together with \RSL, are at the core of 
our new schemes. 

The new attacks we propose 
are of both types: combinatorial and algebraic. 
For all these attacks, we provide a precise analysis of their complexity.

Overall, when all of these new improvements for the RQC scheme
are put together, and their security
evaluated with our different attacks, they enable one to 
gain 50\% in parameter sizes compared to the previous RQC version.
More precisely, we give very competitive parameters, around 11~KBytes,
for RQC schemes with unstructured public key matrices.
This is currently the only scheme with such short parameters
whose security relies solely on pure random instances without any masking
assumptions, contrary to McEliece-like schemes.
At last, when considering the case of Non-Homogeneous errors, our scheme permits 
to reach even smaller parameters. 

\end{abstract}

\section{Introduction}

    \subsubsection*{Background on rank metric code-based cryptography.}
      In the last decade, rank metric code-based cryptography has evolved to become a
      real alternative to traditional code-based cryptography based on the
      Hamming metric. The original scheme based on rank metric was the GPT cryptosystem
      \cite{GPT91}, an adaptation of the McEliece scheme in a rank metric context where
      Gabidulin codes \cite{G85}, a rank metric analogue of Reed-Solomon codes, were the masked
      codes. 
      However, the strong algebraic structure of these codes was successfully 
      exploited for attacking the original GPT cryptosystem and its variants with 
      the Overbeck attack \cite{O05} (see \cite{OTN18} for the latest developments). 
      This situation is similar to the Hamming metric where most of McEliece cryptosystems 
      based on variants of Reed-Solomon codes have been broken.

      Besides the McEliece scheme where a secret code is masked through using permutation, 
      it is possible to generalize the approach by considering public key matrices with trapdoor. 
      Examples of such an approach are NTRU \cite{HPS98} or MDPC \cite{MTSB13} 
      cryptosystems where the masking consists
      in knowing a very small weight vector of the given public matrix. 
      Such an approach was adapted for rank metric through the introduction of 
      LRPC codes \cite{GMRZ13}, a rank metric analogue of MDPC. 

      The security of such type of cryptosystems relies on the general rank decoding problem together with the computational indistinguishability of the public key (a public matrix). 
      The fact that the public matrix is used both for encryption 
      and decryption, permits to obtain very efficient schemes, at the cost of an inversion. 
      It is worth noticing that Loidreau's scheme, which uses homogeneous LRPC matrices in 
      a McEliece context, seems to resist to structural attacks with an homogeneous 
      matrix of sufficiently high enough rank \cite{L17}.

    \subsubsection*{The RQC scheme.} 
      Another approach, proposed by Aleknovich in \cite{A03}, permits to rely solely on random instances of the Syndrome Decoding problem without any masking of a public key. However, such an approach is 
      strongly inefficient in practice; a few years later a more optimized approach 
      was proposed with the HQC scheme \cite{AABBBDGPZ21a}, relying on Quasi-Cyclic codes. It has been generalized to rank metric 
      with the RQC scheme \cite{AABBBBCDGHZ20}. 
      For these schemes, two type of codes are used: a first random double circulant code permits 
      to ensure the security of the scheme when a second public code permits to decode/decrypt
      the ciphertext. In RQC, Gabidulin codes are used as public decryption codes.
      Besides RQC, some other variations were proposed in 
      \cite{GHPT17a,wang2019loong,gaborit2020dual}.
      The main advantage of the RQC cryptosystem compared to the LRPC cryptosystem is the fact 
      that its security reduction is done to random decoding instances whereas the LRPC approach 
      requires another indistinguishability assumption; however, this advantage comes at a price 
      since parameters are larger for RQC than for LRPC.

      The RQC scheme was proposed to NIST Standardization Process with very competitive 
      parameters but algebraic attacks of \cite{BBBGNRT20,BBCGPSTV20}, which were published 
      during the standardization process, had a dreadful impact on RQC parameters 
      so that, in order not to increase too much RQC parameters, the introduction of
      non-homogeneous errors \cite{AABBBBCDGHZ20} permitted to limit the impact of these algebraic attacks. 

      The idea of non-homogeneous errors is to consider errors in three parts of length $n$ 
      such that the error weight is the same for the first two sets, but larger for the third one.
      Such an approach permits to limit the impact of the 
      security reduction of RQC to decoding random $[3n,n]$ codes rather than $[2n,n]$ 
      codes in LRPC cryptosystems. 
      The notion of non-homogeneous error led to the introduction of the Non Homogeneous \RSD problem (\NHRSD)
      and was a first approach to decrease the size of RQC parameters. 
      At this point, it is meaningful to notice that for LRPC and RQC systems, the weight
      of the error to attack is structurally $\mathcal{O}(\sqrt{n})$ (where $n$ is the length of the code), 
      a type of parameters for which algebraic attacks are very efficient.

      Besides the \RSD problem, the \RSL problem which consists in having $N$ syndromes
      whose associated errors share the same support, was introduced
      in \cite{GHPT17a} to construct the RankPKE scheme and later in \cite{wang2019loong}. 
      This problem which generalizes \RSD is meaningful to give more margin in building cryptosystems; 
      it has been recently used to improve on the LRPC schemes \cite{LRPC2022}. 
      It permits, in particular, to increase the weight of the error to decode from $O(\sqrt{n})$ 
      to a weight closer to the Rank Gilbert-Varshamov (RGV) bound; this is (of importance) 
      since for that type of parameters, i.e. close to the RGV bound, algebraic attacks 
      become relatively less efficient than combinatorial attacks.

    \subsubsection*{Attacks and problems in rank metric.} 
      There are two types of attacks in rank metric. Combinatorial 
      attacks which were the first to be introduced in the late 1990's 
      then algebraic attacks ten years later. 
      At first combinatorial attacks were the most efficient ones, but 
      recently and especially for parameters where the error has weight
       $O(\sqrt{n})$ the seminal approaches of \cite{BBBGNRT20,BBCGPSTV20} 
       permitted to have a strong impact on such parameters. 
       Besides the \RSD problem, the \RSL problem was studied in \cite{GHPT17a} 
       and \cite{DT18b} and more recently algebraic attacks 
       were considered in \cite{BB21}. In particular in the definition of 
       the \RSL problem in \cite{GHPT17a} it was shown that giving more 
       than $nr$ syndromes led to a combinatorial attack on the \RSL problem. 
       Moreover, the Non Homogeneous \RSD (\NHRSD) problem was introduced in \cite{AABBBBCDGHZ20}
       in which a first approach for algebraic attack was proposed. 

    \subsubsection*{Contributions.}

      We saw in previous paragraphs, how before the present paper some new problem in rank metric
      had emerged (NHRSD, \RSL)  which permitted to improve on parameters both for RQC and LRPC systems.

      In this paper our contributions are twofold: first we propose new variations on the RQC scheme in order to improve on parameters and second we study in details the new problems on which are based these approaches. All these problems \NHRSD and \RSL appear as natural variations on the \RSD problems and are bound to be the future problems on which will be relying systems in rank metric.

      \medskip
      \noindent \emph{\textbf{New schemes.}}
        The new schemes we propose are based on three types of improvements:

        Our first and main improvement, is the introduction of a new class decodable code, denoted by Augmented Gabidulin codes. These codes exploit the concept of \emph{support erasure} in a
        rank-metric context. Compared to classical Gabidulin codes, the introduction of known support erasure
        permits to decrease the value $m$ down to a value close to the weight of the error,
        whereas $m$ had to be at least twice bigger with classical Gabidulin codes.
        This comes at the cost of a probabilistic decoding; however the decryption failure rate (DFR) can
        be controlled very easily as it is done with LRPC in \cite{GMRZ13,LRPC2022}.

        Second, as for the recent LRPC improvement \cite{LRPC2022} we 
        consider the use of multiple syndromes
        in the RQC scheme. As for LRPC this approach permits to greatly 
        improve the decoding capacity of the RQC scheme
        by increasing the information available for the decryption at a 
        lower cost than directly increasing all parameters.
        This variation implies that the scheme relied on the
        \RSL \cite{GHPT17a,wang2019loong,ABGHZ19} rather than on the \RSD problem.
        As for the new LRPC approach \cite{GMRZ13} this approach permits 
        in particular to increase the weight of the error to decode, so 
        that in practice reaching almost the RGV bound becomes possible (but with larger parameters).

        Third,  like pioneered in \cite{AABBBBCDGHZ20}, we use a variant
        of the \RSD problem by
        considering an error with non-homogeneous weight $(w_1,w_2)$ which is
        the Non-Homogeneous Rank Syndrome Decoding problem (\NHRSD). In short,
        this error contains a part of weight $w_1$ and a part of weight $w_1+w_2$.
        This optimization allows one extra degree of freedom while choosing the target error weights,
        and this has a strong impact on the parameters.

        \medskip
        In conclusion, we propose two types of scheme with very competitive sizes.
        First, \schemeone has parameters similar to MS-LRPC \cite{LRPC2022},
        around 4.5 KBytes for the public key together with the ciphertext; its
        security relies on ideal-codes.
        Second, \schemetwo has parameters a little bit larger than MS-LRPC,
        around 11 KBytes total, this time without any structure; more precisely,
        it relies only on pure random instances of the \RSL problem.
        This is the most conservative security one could expect.
        For both of the aforementioned schemes, one could add a non-homogeneous structure
        in order to shorten the sizes down of 30\%, this corresponds to our scheme:
        NH-\schemeone and NH-\schemetwo.

        The scheme we propose without any ideal structure with small parameters of 11KBytes
        is meaningful, indeed  since it is not proven
        that any ideal structure cannot be used to get faster attack with a quantum computer, scheme without any ideal structure may hence provide a better security. Of course in that case the size of parameters increases a lot but for rank metric our scheme shows that it remains small when for Hamming metric having no additional structure implies very large public key
        (see McEliece scheme for instance). Moreover our scheme does not necessitate any supplementary indistinguishability assumption. 
        Moreover, our schemes compare very well with other code-based schemes.

      \bigskip
      \noindent \emph{\textbf{New attacks and analysis.}}
        We saw that our new improvements on RQC relied on recent problems for rank metric, 
        namely the \NHRSD et \RSL problems (and also a combination of the two latter problems the \NHRSL problem).
        Although these schemes have begun to be considered
        we go deeper in their study by proposing new attacks and adaptation of known attacks for the security evaluation of these problems. The motivation comes both from the general interest of these problems for rank based cryptography and for the new schemes that we introduce in this paper.

        More precisely, recall that an \RSL$(m,n,k,r,N)$ instance is like a 
        rank syndrome decoding instance
        of parameters $(m,n,k,r)$ where $N$ instead of 1 syndromes are given,
        and all their associated errors
        share the same support.
        The security of \RSL is inherent to the value of $N$, and it is known since
        \cite{GHPT17a} that the problem can be solved in polynomial time as long as
        $N>nr$. Our contributions are then the following:
        \begin{itemize}
          \item With our new combinatorial attack against \RSL, first we improve on the most
          recent algebraic attack for some instances; most importantly, we improve the aforementioned
          bound, unchanged since 2017 \cite{GHPT17a}, showing that \RSL becomes polynomial
          as long as $N>kr\frac{m}{m-r}$.
          \item We also propose the first combinatorial attack against \NHRSD, together
          with a precise complexity analysis of the algebraic attack, still against \NHRSD,
          described in \cite{AABBBBCDGHZ20}.
          \item Finally, we propose an attack against \NHRSL. That it is
          to say that we were able to take advantage of two structure in the same attack: the fact that
          the error is non-homogeneous and that one is given several syndromes.
        \end{itemize} 

\section{Preliminaries}

  \subsection{Coding theory and rank metric}

    Let $q$ be a prime power, let $m$ a positive integer, let $\ff{q^m}$ an extension of degree $m$ of $\ff{q}$ and let $\mat\beta := (\beta_1,\dots,\beta_m)$ be an $\ff{q}$-basis of $\ff{q^m}$. 
    Any vector in $\ff{q^m}^n$ can naturally be viewed as a matrix in $\ff{q}^{m \times n}$ by expressing its coordinates in $\mat\beta$.

    \begin{definition}[Rank weight]
    Let $\xv = (x_1,\ldots,x_n) \in \ff{q^m}^{n}$ be a vector. The rank weight of $\xv$ denoted $\norme{\xv}$ is defined as the rank of the matrix $\Mat(\xv) := (x_{ij})_{i,j} \in \ff{q}^{m \times n}$ where $x_j = \beta_1 x_{1j} + \dots + \beta_m x_{mj}$ for $j \in \{1..n\}$. 
    The set of vectors of weight $w$ in $\Fqm^n$ is denoted $\mathcal{S}_{w}^{n}(\Fqm)$.
    \end{definition}

    \begin{definition}[Support]
    The support of $\xv \in \ff{q^m}^{n}$ is the $\ff{q}$-linear space generated by the coordinates of $\xv$, \emph{i.e.} $\text{Supp}(\xv) := \langle x_1,\dots,x_n \rangle_{\ff{q}}$.
    It follows from the definition that $\norme{\xv} = \dim_{\ff{q}}{\left(\text{Supp}(\xv) \right)}$.
    \end{definition}

    These notions can be extended to matrices. The support of a matrix $\mathbf{M} \in \fqm^{r \times c}$ denoted $\Supp(\mathbf{M})$ is the $\fq$-vector space spanned by all its $r \times c$ entries, and the rank weight $\norme{\mathbf{M}}$ is defined as the dimension of this support. Note that we always have $\norme{\mathbf{M}} \leq \rank{(\mathbf{M})}$, for example $\norme{\mathbf{I}_n} = 1$ while $\rank{(\mathbf{I}_n)} = n$.

    \begin{definition}[$\ff{q^m}$-linear code]
    An $\ff{q^m}$-linear code $\mathcal{C}$ of length $n$ and dimension $k$ is an $\Fqm$-linear subspace of $\Fqm^n$ of dimension $k$. We say that it has parameters $[n,k]_{q^m}$. A generator matrix for $\mathcal{C}$ is a full-rank matrix $\Gm \in \ff{q^m}^{k \times n}$ such that
    $\mathcal{C} = \left\lbrace \mv \Gm,~\mv \in \ff{q^m}^k\right\rbrace.$
    A parity-check matrix is a full-rank matrix $\Hm \in \ff{q^m}^{(n-k) \times n}$ such that
    $\mathcal{C} = \left\lbrace \xv \in \ff{q^m}^n,~\Hm \trsp{\xv} = 0\right\rbrace.$
    Finally, the rowspace of $\Hm$ is a basis of the dual code $\mathcal{C}^{\perp}$.
    \end{definition}

    The use of $\ff{q^m}$-linear codes instead of standard $\ff{q}$-linear codes permits to obtain a more compact description for the public key in code-based cryptosystems.
    Another classical way to reduce the keysize is to use some type of cyclic structure, which leads to the notion of ideal codes.
    Let $P \in \ff{q}[X]$ denote a polynomial of degree $n$.
    The linear map $\uv := (u_0,\dots,u_{n-1}) \mapsto \uv(X) := \sum_{i=0}^{n-1}u_i X^i$ is a vector space isomorphism between $\ff{q^m}^{n}$ and $\ff{q^m}[X] / \langle P \rangle$, and we use it to define a product between two elements $\uv$ and $\vv$ in $\ff{q^m}^n$ via $\uv \cdot_{P} \vv := \uv(X)\vv(X) \text{ mod }P$.
    Note that we have 
    \begin{align*}
    \uv \cdot \vv = \left( \sum_{i=0}^{n-1}u_i X^{i} \right) \vv(X) \text{ mod }P 
    = \sum_{i=0}^{n-1}u_i \left( X^{i} \vv(X) \text{ mod }P \right),
    \end{align*}
    so that the product by $\vv \in \ff{q^m}^n$ can be seen as a matrix-vector product by the so-called ideal matrix generated by $\xv$ and $P$.

    \begin{definition}[Ideal matrix]
    Let $P \in \ff{q}[X]$ a polynomial of degree $n$ and let $\vv \in \ff{q^m}^n$. 
    The ideal matrix generated by $\vv$ and $P$, noted $\mathcal{IM}_{P}(\mathbf{v})$, 
    is the $n \times n$ matrix, with entries in $\fqm$, and whose rows are the following:
    $\mathbf{v}(X) \mod P$, $X \mathbf{v}(X) \mod P$, $\hdots$, 
    $X^{n-1} \mathbf{v}(X) \mod P$.

    For conciseness, we use the notation $\mathcal{IM}(\mathbf{v})$ since there will be no ambiguity in the choice of $P$ in the paper. 
    \end{definition}

    \noindent One can see that $\uv \cdot \vv = \uv\mathcal{IM}(\mathbf{v}) = \vv \mathcal{IM}(\mathbf{u}) = \vv \cdot \uv$. 
    An ideal code $\mathcal{C}$ of parameters $[sn,tn]_{q^m}$ is an $\ff{q^m}$-linear code which admits a generating matrix made of $s \times t$ ideal matrix blocks in $\ff{q^m}^{n \times n}$. 
    A crucial point regarding the choice of the modulus $P$ (see \cite[Lemma 1]{AABBBBCDGHZ20}) is that if $P \in \ff{q}[X]$ is \emph{irreducible} of degree $n$ and if $n$ and $m$ are \emph{prime}, then such a code $\mathcal{C}$ always admits a \emph{systematic} generator matrix made of ideal blocks. 
    Hereafter, we only consider $t = 1$.

    \begin{definition}[Ideal codes]
    Let $P \in \ff{q}[X]$ a polynomial of degree $n$. An $[ns, n]_{q^m}$-code $\C$ is an ideal code if it has a generator matrix of the form 
    $\Gm =  \begin{pmatrix}
    \mat{I}_{n} ~ & \mathcal{IM}(\mathbf{g}_1) && \dots && \mathcal{IM}(\mathbf{g}_{s-1}) \\
    \end{pmatrix} \in \ff{q^m}^{n \times ns}$, where $\mathbf{g}_{i} \in \ff{q^m}^n$ 
    for $1 \leq i \leq s-1$. Similarly, $\C$ 
    is an ideal code if it admits a parity-check matrix of the form

    $$\Hm =  \begin{pmatrix}
    &~ \trsp{\mathcal{IM}(\mathbf{h}_1)} \\
    \mat{I}_{n(s-1)}& \vdots \\
    &~ \trsp{\mathcal{IM}(\mathbf{h}_{s-1})}
    \end{pmatrix} \in \ff{q^m}^{n(s-1) \times ns}.$$
    \end{definition}

  \subsection{Gabidulin codes}

    Gabidulin codes were introduced by Gabidulin in 1985 \cite{G85}. These codes can be seen as the rank metric analogue of Reed-Solomon codes \cite{ReedSolomon60}, where standard polynomials are replaced by $q$-polynomials (also called Ore polynomials or linearized polynomials).

    \begin{definition}[$q$-polynomials]\label{def:q-polynomes}
    The set of $q$-polynomials over $\Fqm$ is the set of polynomials with the following shape:
    \[ \left\lbrace P(X) = \sum_{i = 0}^{r} p_i X^{q^i}, \text{ with } p_i
    \in \Fqm \text{ and } p_r \neq 0 \right\rbrace. \]
    The $q$-degree
    of a $q$-polynomial $P$ is defined as $\deg_q(P) = r$.
    \end{definition}

    \begin{definition}[Ring structure]
    The set of $q$-polynomials over $\Fqm$ is a non-commutative ring 
    for $(+,\circ)$, where $\circ$ is the composition of $\ff{q}$-linear endomorphisms.

    \end{definition}
    Due to their structure, the $q$-polynomials are inherently related to decoding problems in the rank metric as stated by the following
    propositions.

    \begin{theorem}[\cite{ore1933special}]
    Any $\Fq$-subspace of $\Fqm$ of dimension $r$ is the set of the roots of a unique monic $q$-polynomial $P$ such that $\deg_q(P) = r$.
    \end{theorem}

    \begin{corollary}
    Let $\mathbf{x}=\left(x_1,x_2,\dots,x_n\right) \in \mathbb{F}_{q^m}^n$ and $V$ be the monic $q$-polynomial of smallest $q$-degree such that $V(x_i) = 0$ for $1 \leq i \leq n$, then $\norme{\mathbf{x}} = r$ if and only if $\deg_q(V) = r$.
    \end{corollary}

    Finally, Gabidulin codes can be seen as the evaluation of $q$-polynomials of
    bounded degree on the coordinates of a fixed vector over $\Fqm$.  

    \begin{definition}[Gabidulin codes]\label{def:code_gabidulin}
    Let $k,n,m \in \N$ such that $k \leqslant n \leqslant m$ and let
    $\gv = (g_1,\dots,g_n)$ be an-$\Fq$ linearly independent family of
    elements of $\Fqm$. The {\em Gabidulin code} $\mathcal{G}_{\gv}(n,k,m)$ is
    the code of parameters $[n,k]_{q^m}$ defined by
    \[ 
      \mathcal{G}_{\gv}(n,k,m) := \left\lbrace P(\gv),~\deg_q (P) < k \right\rbrace,
      \text{where $P(\gv) := (P(g_1), \ldots, P(g_v))$.}
    \]
    
    \end{definition}

  \subsection{Hard problems in rank-based cryptography}
 
    As in the Hamming metric, the main source of computational hardness for rank-based cryptosystems is a decoding problem. More precisely, it is the decoding problem in the rank metric setting restricted to $\Fqm$-linear codes which is called the Rank Syndrome Decoding Problem ($\RSD$).

    \begin{probe}[\RSD Problem, Search]\label{problema:rsd}
      Given $\Hm \in \ff{q^m}^{(n-k) \times n}$, 
      a full rank parity-check matrix for a random $\ff{q^m}$-linear code $\mathcal{C}$, 
      an integer $w \in \mathbb{N}$ and a syndrome $\sv \in \Fqm^{n-k}$, 
      the Rank Syndrome Decoding problem $\RSD(m,n,k,w)$ asks to find 
      $\ev \in \ff{q^m}^n$ such that $\norme{\ev} = w$ and $\Hm \trsp{\ev} = \trsp{\sv}$.
    \end{probe}

    The decision version is denoted by \DRSD. Even if \RSD is not known to be NP-complete, there exists a randomized reduction from \RSD to an NP-complete problem, namely to decoding in the Hamming metric \cite{GZ14}. Also, the average number of solutions for a fixed weight $w$ is given by the following Gilbert-Varshamov bound for the rank metric:
    \begin{definition}[Rank Gilbert-Varshamov bound]
    The Gilbert-Varshamov bound $w_{GV}(q,m,n,k)$ for $\ff{q^m}$-linear codes 
    of length $n$ and dimension $k$ in the rank metric is defined as the smallest 
    positive integer $t$ such that 
    $q^{m(n-k)} \leq B_t$, 
    where 
    $B_t := \sum_{j=0}^{t}\left(\prod_{\ell =0}^{j-1}(q^n - q^{\ell}) \right) \binom{m}{j}_q$ 
    is the size of the ball of radius $t$ in the rank metric.
    \end{definition}

    In other words, it means that, with overwhelming probability, as long as $w \leq w_{GV}(q,m,n,k)$, 
    a random \RSD instance will have at most a unique solution. In this paper, we also focus on a slightly less standard assumption which is the \NHRSD problem. This \RSD variant was proposed in the Second Round update of RQC \cite{AABBBBCDGHZ20} in order to mitigate the impact
    of the recent algebraic \RSD attacks \cite{BBBGNRT20,BBCGPSTV20} on the choice of the parameters. In \NHRSD, the error $\ev$ is no longer a random low weight vector but instead a vector with a \emph{non-homogeneous} weight:

    \begin{probe}[NHRSD Problem, Search]\label{problema:nhrsd}
      Given $\Hm \in \ff{q^m}^{(n+n_1) \times (2n+n_1)}$, a full rank 
      parity-check matrix of a random $\ff{q^m}$-linear code $\mathcal{C}$ 
      of parameters $[2n+n_1,n]_{q^m}$, integers $(w_1,w_2) \in \mathbb{N}^2$, 
      and a \ syndrome \ $\sv \in \Fqm^{n+n_1}$, \ the Non-Homogeneous Rank Syndrome 
      Decoding problem $\NHRSD(m,n,n_1,w_1,w_2)$ asks to find a vector  
      $\ev = (\ev_1,\ev_2,\ev_3) \in \ff{q^m}^{2n+n_1}$ 
      such that 
      $\norme{(\ev_1,\ev_3)} = w_1,~\norme{\ev_2} = w_1+w_2$, 
      $\Supp(\ev_1,\ev_3) \subset \Supp(\ev_2)$, 
      and such that $\Hm \trsp{\ev} = \trsp{\sv}$.
    \end{probe}
    We denote by $\DNHRSD$ the corresponding decisional version. Note that Definition \ref{problema:nhrsd} is slightly more general than the one of \cite{AABBBBCDGHZ20} where it is assumed that $n=n_1$. 
    Finally, recall that one of our improvements on the RQC scheme uses multiple syndromes 
    which are correlated since they correspond to errors which share the same support. 
    This formulation exactly corresponds to the definition of the Rank Support Learning problem 
    ($\RSL$ and $\DRSL$ for the decision version). This problem can be seen as the rank metric analogue of the Support-Learning problem in the Hamming metric \cite{KKS97,OT11}.
    \begin{probe}[$\RSL$ Problem, Search]\label{problema:rsl}	
      Given $(\mat{H},\mat{H\trsp{E}})$, where $\mat{H} \in \ff{q^m}^{(n-k) \times n}$ 
      is of full-rank, and $\mat{E} \in \ff{q^m}^{N \times n}$ has all its entries lying 
      in a subspace $\mathcal{V} \subset \ff{q^m}$ of dimension $w \in \mathbb{N}$, 
      the Rank Support Learning problem $\RSL(m,n,k,w,N)$ 
      asks to find the secret subspace $\mathcal{V}$.
    \end{probe}

    \RSL may enable the construction of more advanced cryptographic 
    primitives in the rank metric. It was introduced in \cite{GHPT17a}, and 
    is at the core of the Durandal signature scheme \cite{ABGHZ19}, 
    and the recent Multi-LRPC proposal \cite{LRPC2022}. 
    Naturally, it is possible to somehow combine the error distributions 
    from Problems \ref{problema:nhrsd} and \ref{problema:rsl}.

    \begin{probe}[$\NHRSL$ Problem, Search]\label{problema:nhrsl}
    	Given $(\mat{H},\mat{H\trsp{E}})$, 
      where $\mat{H}$ is a $(n+n_1) \times (2n+n_1)$ matrix of full rank, 
      and $\mat{E} \in \ff{q^m}^{N \times (2n+n_1)}$ such that 
      $\ev_i = \mat{E}_{i,*} = (\ev_{i,1},\ev_{i,2}, \allowbreak\ev_{i,3}) \in \ff{q^m}^{(2n+n_1)}$, 
      $\norme{(\ev_{i,1},\ev_{i,3})} = w_1,~\norme{\ev_{i,2}} = w_1+w_2$, 
      and such that the supports 
      $\mathcal{V} := \Supp(\ev_{i,1},\ev_{i,3}) \subset \mathcal{W} := \Supp(\ev_{i,2})$ 
      are independent of $i$; 
      the Non-Homogeneous Rank Support Learning problem 
      $\NHRSL(m, n, n_1, w_1, w_2, N)$ asks to find the secret subspaces $\mathcal{V}$ and $\mathcal{W}$.
    \end{probe}

    Finally, both classic RQC and our \schemeone proposal involve ideal codes, 
    so that we have to consider the ideal versions of Problems 
    \ref{problema:rsd}, \ref{problema:nhrsd}, \ref{problema:rsl}, 
    and \ref{problema:nhrsl} (denoted by, 
    respectively,  $\IRSD$, $\NHIRSD$, $\IRSL$, and $\NHIRSL$). 
    For the sake of conciseness, we do not give a formal definition 
    of these ideal variants. 

  \subsection{RQC scheme}

    On Figure \ref{fig:scheme-pke} we briefly recall 
    the classical RQC scheme \cite{AABBBBCDGHZ20}, 
    for which one needs the following notation:
     \begin{align*}
      \mathcal{S}_{w,1}^n(\Fqm) = & \ \{ \xv \in \Fqm^n : \norme{\xv} = w, 1 \in  \Supp(\xv)\},\\
      \mathcal{S}_{(w_1,w_2)}^{3n}(\Fqm) = & \ \{ \xv = (\xv_1,\xv_2,\xv_3) \in \Fqm^{3n} : \norme{(\xv_1,\xv_3)} = w_1, \norme{\xv_2} = w_1+w_2, \\
        & \quad \Supp(\xv_1,\xv_3) \subset \Supp(\xv_2)\}.
      \end{align*}

      \vspace{-2\baselineskip}

      \begin{figure}[!ht]
        \centering
        \fbox{
        \begin{minipage}{.965\textwidth}

          \underline{\Setup$(1^\lambda)$}: 
          Generates and outputs \param{} = $(n, k, \delta, w, w_1,w_2,P)$ 
          where $P \in \Fq[X]$ is an irreducible polynomial of degree $n$.

          \vspace{0.5\baselineskip}
          \underline{\KeyGen$(\param)$}: Samples $\mathbf{h} \lar \Fqm^n$, 
          $\mathbf{g} \lar \mathcal{S}_{n}^{n}(\Fqm)$ and 
          $(\mathbf{x}, \mathbf{y}) \lar \mathcal{S}_{w,1}^{2n}(\Fqm)$, 
          computes the generator matrix $\mathbf{G}\in \mathbb{F}_{q^m}^{k\times n}$ 
          of a code 
          $\mathcal{C}$, sets $\pk = \left(\mathbf{g}, \mathbf{h}, 
          \mathbf{s} = \mathbf{x+h\cdot y} \mod P \right)$ and 
          $\sk= \left(\mathbf{x}, \mathbf{y}\right)$, returns $(\pk, \sk)$.

          \vspace{0.5\baselineskip}
          \underline{\Encrypt$(\pk, \mathbf{m}, \theta)$}: 
          Uses randomness $\theta$ to generate $(\mathbf{r}_1, \ev, \mathbf{r}_2) 
          \lar \mathcal{S}_{(w_1,w_2)}^{3n}(\Fqm)$, sets $\mathbf{u} = 
          \mathbf{r}_1+\mathbf{h}\cdot\mathbf{r}_2 \mod P$ and $\mathbf{v} = 
          \mathbf{m}\mathbf{G} + \mathbf{s\cdot r}_2 + \mathbf{e} \mod P$, 
          returns $\mathbf{c} = \left(\mathbf{u}, \mathbf{v}\right)$.

          \vspace{0.5\baselineskip}
          \underline{\Decrypt$(\sk, \mathbf{c})$}: 
          Returns $\mathcal{C}$.$\mathsf{Decode}(\mathbf{v}-\mathbf{u\cdot y} \mod P)$.
        \end{minipage}
        }
        \caption{\label{fig:scheme-pke}Description the RQC PKE scheme.}
      \end{figure}

\section{Augmented Gabidulin code: a new family of efficiently decodable codes for cryptography}

  In what follows, we introduce a new family of efficiently decodable 
  codes, namely Augmented Gabidulin codes. 
  The main idea behind these codes is to add a sequence of zeros at the end 
  of the Gabidulin codes; by doing this, one directly gets elements of the support 
  of the error, which correspond to \emph{support erasure} in a rank metric context. 
  The decoding of this code corresponds to the decoding of a classical Gabidulin code 
  to which support erasures are added. In practice, this approach permits to decrease the 
  size of $m$ at the cost of having a probabilistic decoding. 
  The probability of decoding failure can then be easily controlled at
  the cost of sacrifying only a few support erasures; indeed, for these codes, 
  the decoding failure probability decreases exponentially fast, with a quadratic exponent, 
  see Equation (\ref{proba_augmented_gabi}).

  This approach is then especially suitable in the case where many errors have to 
  be corrected which exactly corresponds to our case where the code we want to decode has a very low rate. 
  
  In what follows, we give a definition of augmented Gabidulin codes, and 
  for didactic purpose, in Proposition \ref{prop:decoding_capacity}, 
  we recall a simple and natural way 
  to decode Gabidulin code with support erasures and we give their decoding 
  failure rate. 
  For other or more efficient approaches, the reader may refer to 
  \cite{augotLoidreau,couvreurBombar,gabidulin2008}.

  Notice that this type of approach (adding zeros) is not relevant in Hamming metric since 
  the errors are independent in a classical noisy canal, whereas in rank metric, errors 
  located on different coordinates are linked since they share the same support. 
  
  \begin{definition}[Augmented Gabidulin codes]\label{def:aug_code_gabidulin}
  Let $(k,n,n',m) \in \N^4$ such that $k \le n' < m < n$. Let
  $\gv = (g_1,\dots,g_{n'})$ be an $\Fq$- linearly independent family of $n'$ 
  elements of $\Fqm$ and let $\overline{\gv}$ be the vector of length $n$ which 
  is equal to $\gv$ padded with $n-n'$ extra zeros on the right. The {\em Augmented Gabidulin code} 
  $\mathcal{G}^+_{\overline{\gv}}(n,n',k,m)$ is
  the code of parameters $[n,k]_{q^m}$ defined by
  \[ \mathcal{G}^+_{\overline{\gv}}(n,n',k,m) := \left\lbrace P(\overline{\gv}),~\deg_q (P) < k \right\rbrace,\]
  where $P(\overline{\gv}) :=
  (P(g_1), \ldots, P(g_{n'}),0,\dots,0)$.

  \end{definition}

  \begin{proposition}[Decoding capacity of Augmented Gabidulin codes]\label{prop:decoding_capacity}
    Let $\mathcal{G}^+_{\overline{\gv}}(n,n',k,m)$ be an augmented 
    Gabidulin code, and let $$\varepsilon \in \{1,2,\hdots,\operatorname{min}(n-n',n'-k)\}$$ 
    be the dimension of the vector space 
    generated by the support erasures. 
    
    Then, $\mathcal{G}^+_{\overline{\gv}}(n,n',k,m)$ can uniquely decode an error of 
    rank weight up to 
    $$t:=\left\lfloor \frac{n'-k+\varepsilon}{2} \right\rfloor.$$ 
  \end{proposition}

  \begin{proof}
    The minimal distance of $\mathcal{G}^+_{\overline{\gv}}(n,n',k,m)$ is clearly 
    $d=n'-k+1$ since it is made of a Gabiudlin code augmented with zeros. 

    Let $x=c_1+e_1$ be a noisy codeword where $c_1 \in \mathcal{G}^+_{\overline{\gv}}$, 
    and $\norme{e_1}\le t$.

    Let us assume that $x$ is not uniquely decodable to find a contradiction. 
    
    If $x$ is not uniquely decodable, it means that there exists $c_2$ in 
    $\mathcal{G}^+_{\overline{\gv}}$ such that $c_2 \ne c_1$ and $x=c_2+e_2$ 
    where $\norme{e_2}\le t$. 

    Recall that we assume that one knows support erasures which span a vector space 
    of dimension $\varepsilon$. These support erasures come from the $n-n'$ last 
    coordinates of the code $\mathcal{G}^+_{\overline{\gv}}$, thus these support elements 
    are common to $e_1$ and $e_2$. 
    Since $\operatorname{Supp}(e_1)$ and $\operatorname{Supp}(e_2)$ share $\varepsilon$ 
    elements, one has that 
    $$\operatorname{d}(e_1,e_2) \le 2(t-\varepsilon)+\varepsilon = 2t-\varepsilon \le n'-k.$$
    
    Since $x=c_1+e_1=c_2+e_2$, one clearly has that 
    $\operatorname{d}(c_1,c_2)=\operatorname{d}(e_1,e_2)$, thus $\operatorname{d}(c_1,c_2)\le n'-k$, 
    which is a contradiction. 

    Thus, $\mathcal{G}^+_{\overline{\gv}}$ can uniquely decode errors of rank weight up to 
    $t:=\left\lfloor \frac{n'-k+\varepsilon}{2} \right\rfloor$.

    Finally, the condition $1 \le \varepsilon \le \operatorname{min}(n-n',n'-k)$ 
    comes from the fact that the dimension 
    of the vector space spanned by support erasures can not exceed the maximum 
    rank weight of the error nor the number of zero coordinates of the 
    augmented Gabidulin code; in other words, on one hand $\varepsilon$ is clearly smaller than 
    $n-n'$, and on the other hand 
    \begin{align*}
      \varepsilon   \le \left\lfloor \frac{n'-k+\varepsilon}{2} \right\rfloor 
        & \Longrightarrow 2\varepsilon   \le n'-k+\varepsilon \\ 
        & \Longrightarrow \varepsilon   \le n'-k.
    \end{align*}
    \qed
  \end{proof}

  \begin{proposition}[Decoding Algorithm for Augmented Gabidulin codes]\label{prop:decoding_algo}
    Let $\mathcal{G}^+_{\overline{\gv}}(n,n',k,m)$ be an augmented 
    Gabidulin code, and let $$\varepsilon \in \{1,2,\hdots,\operatorname{min}(n-n',n'-k)\}$$ 
    be the dimension of the vector space 
    generated by the support erasures. 

    This code benefits from an 
    efficient decoding algorithm correcting errors of rank weight up to 
    $\delta := \left\lfloor \frac{n'-k+\varepsilon}{2} \right\rfloor$ 
    with a decryption failure rate (DFR) of 
    \begin{equation}\label{proba_augmented_gabi}    
    1-\frac{1}{\delta(n-n')}\sum_{i=\varepsilon}^{\delta}
    \prod_{j=0}^{\varepsilon-1}
    \frac{(q^\delta-q^j)(q^{n-n'}-q^j)}{q^\varepsilon-q^j}. 
    \end{equation}
  \end{proposition}
  \begin{proof}
    The proof gives the decoding algorithm. One is given a noisy encoded word 
  	$\overline{\yv}=\overline{\cv}+\overline{\ev} \in \fqm^n$ 
  	where $\overline{\cv}:=\xv\Gm$ belongs to 
  	$\mathcal{G}^+_{\overline{\gv}}(n,n',k,m)$ and 
  	$\norme{\ev} \le \delta$.

  	\paragraph{Step 1: recovering a part of the error support.} 
    By construction we have $\overline{\cv}=(*|0\hdots0)$, so that
  	the last $n-n'$ coordinates of $\yv$ are exactly the last  
  	coefficients of $\overline{\ev}$. 
  	Thus, one may use these coefficients to 
  	recover $\varepsilon$ elements in  $E := \text{Supp}(\overline{\ev})$. This will be doable as long as these 
  	$n-n'$ coefficients contain at least $\varepsilon$ linearly independent ones. The converse probability is the probability that  
  	a random $\delta \times (n-n')$ matrix with coefficients in $\fq$ has rank less than 
  	$\varepsilon$. This yields to the probability given by Equation (\ref{proba_augmented_gabi}).

  	\paragraph{Step 2: recovering $\mathbf{\overline{c}}$.} Assume now that $\varepsilon$ elements in the support of $\overline{\ev}$ are known and let $E_2$ be the vector space spanned by these elements. In what follows, we focus on the first $n'$ coordinates of 
  	$\overline{\yv},\overline{\cv},$ and $\overline{\ev}$ which are denoted by
  	$\yv, \cv$ and $\ev$ respectively. 
    By definition of $\mathcal{G}^+_{\overline{\gv}}(n,n',k,m)$, there exists a $q$-polynomial $P$ of $q$-degree at most $k-1$ such that for $1 \leq i \leq n'$:
  	\begin{equation}\label{eq:decodage_1}
  		y_i=P(g_i)+e_i.
  	\end{equation} 
  	Let also $V$ and $V_2$ be the unique monic $q$-polynomials of $q$-degree $\delta$ and $\varepsilon$ which vanish on the vector spaces $E$ and $E_2$ respectively. The ring of $q$-polynomials being left Euclidean, there exists a unique monic $q$-polynomial $W$ of degree $\delta-\varepsilon$ such that 
  	$	V=W \circ V_2$. As $E_2$ is known, one can easily build the q-polynomial $V_2$, for instance using the iterative process described in \cite{ore1933special,L06}. Evaluating $V$ at both sides of Equation (\ref{eq:decodage_1}), one gets 
  	$V(y_i)=(V \circ P)(g_i)+V(e_i) = V \circ P(g_i)$. This secret polynomial can be written symbolically using $\delta-\varepsilon$ unknowns in $\fqm$, and similarly we view $R:=V \circ P$ as a $q$-polynomial of $q$-degree $k-1+\delta$ with unknown coefficients. Thus, we can derive a linear equation containing $k+2\delta-\varepsilon$ unknowns in $\fqm$ from 
  	\begin{equation}\label{eq:decodage_2}
  		V(y_i)=R(g_i),
  	\end{equation}
  	and the same goes for any $i \in \{1,2,\hdots n'\}$. Overall, this gives a linear system with 
  	$n'$ equations in $k+2\delta-\varepsilon$ variables. 
    This linear system has more equations than unknowns as long as 
    $
      \delta \le \left\lfloor \frac{n'-k+\varepsilon}{2} \right\rfloor,
    $ 
    which is the case by assumption. Moreover, this system has a unique 
    solution by Proposition \ref{prop:decoding_capacity}. 
    This means that exactly $k+2\delta-\varepsilon$ equations are linearly 
    independent, thus one can solve the system to 
    recover $V$ and $R$, so one finally gets $P$.
    \qed
  \end{proof}

\section{New Rank-based Encryption Schemes}

\subsection{\schemeone scheme} \label{sec:scheme-multi}

Our new encryption scheme denoted \schemeone stands for RQC with multiple syndromes.
Indeed, it uses several syndromes $\mathbf{U}$ and $\mathbf{V}$ which differ from the original RQC proposal that relies on unique syndromes $\mathbf{u}$ and $\mathbf{v}$.
As a consequence, our new scheme is based on the $\IRSL$ problem which can be seen as a generalization of the $\IRSD$ problem used by RQC.

\vspace{\baselineskip}
\noindent \textbf{Notations.} We start by introducing several sets and operators required to define the \schemeone scheme.
Let $\mathcal{S}_{w,1}^{2n}(\Fqm)$ and $\mathcal{S}_{(w_1,w_2)}^{n_2 \times 3 n_1}(\Fqm)$ be defined as:
\begin{align*}
  \mathcal{S}_{w,1}^{2n}(\Fqm) = & \ \{ \xv = (\xv_1, \xv_2) \in \Fqm^{2n} \,\big|\, \norme{\xv} = w, \ 1 \in  \Supp(\xv)\}, \\
  \mathcal{S}_{(w_1,w_2)}^{n_2 \times 3 n_1}(\Fqm) = & \ \{ 
      \mathbf{X} = (\mathbf{X}_1,\mathbf{X}_2,\mathbf{X}_3) \in \Fqm^{n_2 \times 3 n_1} \,\big|\,
      \norme{(\mathbf{X}_1,\mathbf{X}_3)} = w_1, \\
      & \norme{\mathbf{X}_2} = w_1+w_2, ~ \Supp(\mathbf{X}_1,\mathbf{X}_3) \subset \Supp(\mathbf{X}_2)\}.
\end{align*}

\noindent Let $n_1, n_2$ be positive integers such that $n=n_1 \times n_2$,  
for a vector $\mathbf{v} \in \Fqm^{n_2}$ and a matrix 
$\mathbf{M} \in \Fqm^{n_2 \times n_1}$ whose columns are 
labelled $\mathbf{M}_1, \hdots, \mathbf{M}_{n_1}$, 
we extend the aforementioned 
dot product such that:
\vspace{-0.25\baselineskip}
\begin{multline*}
\mathbf{v}\cdot\mathbf{M} =
\left(
  (\mathbf{v}\cdot\mathbf{M}^{\intercal}_1 \mod P)^\intercal, \ 
\hdots, \ 
  (\mathbf{v}\cdot\mathbf{M}^{\intercal}_{n_1} \mod P)^\intercal
\right) \in \Fqm^{n_2 \times n_1},
\end{multline*}

\noindent Let $\mathbf{v} = (\mathbf{v}_1,\hdots,\mathbf{v}_{n_1}) \in \Fqm^n$ 
with $\mathbf{v}_i \in \Fqm^{n_2} ~ \forall i \in \{1,\hdots,n_1\}$,
the $\fold()$ procedure turns the vector $\vv$ into a $n_2 \times n_1$ matrix 
$\fold(\mathbf{v})= (\mathbf{v}_1^\intercal,\hdots,
\mathbf{v}_{n_1}^\intercal) \in \Fqm^{n_2 \times n_1}.$
The procedure $\unfold()$ is naturally defined as the converse 
of $\fold()$.

\vspace{\baselineskip}
\noindent \textbf{Protocol.}
The \schemeone is described on Figure~\ref{fig:scheme-one}. 
It relies on two codes namely an augmented Gabidulin code 
$\mathcal{G}^+_{\overline{\gv}}(n, n', k, m)$ that can correct up to 
$\delta := \left\lfloor \frac{n'-k+\varepsilon}{2} \right\rfloor$
errors using the efficient decoding algorithm 
$\mathcal{G}^+_{\overline{\gv}}.\mathsf{Decode}(.)$ as well as
a random ideal $[2n_2,n_2]_{\fqm}$-code with 
parity check matrix $(\mathbf{I} ~~ \mathcal{IM}(\mathbf{h}))$.
The correctness of the protocol follows from:
\begin{align*}
\mathbf{V}-\mathbf{y}\cdot \mathbf{U} &= 
%\fold(\mathbf{m}\mathbf{G}) + 
%\mathbf{s} \cdot\mathbf{R}_2 + \mathbf{E} -
%\mathbf{y} \cdot (\mathbf{R}_1 + \mathbf{h} \cdot \mathbf{R}_2) \\
%&= 
\fold(\mathbf{m}\mathbf{G}) + 
(\mathbf{x} + \mathbf{h} \cdot \mathbf{y}) \cdot \mathbf{R}_2 + 
\mathbf{E} - 
\mathbf{y} \cdot (\mathbf{R}_1 + \mathbf{h} \cdot \mathbf{R}_2) \\
&= 
\fold(\mathbf{m}\mathbf{G})
+ \mathbf{x} \cdot \mathbf{R}_2
- \mathbf{y} \cdot \mathbf{R}_1
+ \mathbf{E}.
\end{align*}

\noindent As a consequence, $\unfold\left(\mathbf{V}-\mathbf{y}\cdot\mathbf{U}\right)=
\mathbf{m}\mathbf{G}+
\unfold\left(
\mathbf{x} \cdot \mathbf{R}_2
- \mathbf{y} \cdot \mathbf{R}_1
+ \mathbf{E} \right) \in \Fqm^n
$
which means that $
\mathcal{G}_{\gv}.\mathsf{Decode}\left(
\unfold\left(
\mathbf{V}-
\mathbf{y}\cdot\mathbf{U}
\right)\right)=\mathbf{m}
$
as long as: \[\norme{ \,
\unfold\left(
\mathbf{x} \cdot \mathbf{R}_2
- \mathbf{y} \cdot \mathbf{R}_1
+ \mathbf{E} \right)
} \le \delta.\]

\vspace{-1.5\baselineskip}

\begin{figure}[!ht]
\centering
\fbox{
\begin{minipage}{.965\textwidth}

\underline{\Setup$(1^\lambda)$} \\[0.25\baselineskip]
Generate and output the parameters \param{} = 
$(n', n_1, n_2, k, \epsilon, \delta, w, w_1,w_2,P)$ where 
$P~\in\Fq[X]$ is an irreducible polynomial of degree $n_2$.

\vspace{\baselineskip}
\underline{\KeyGen$(\param)$}: \\[0.25\baselineskip]
Sample $\mathbf{g} \lar \mathcal{S}^{n'}_{n'}(\Fqm)$, $\mathbf{h} \lar \Fqm^{n_2}$ and 
$(\mathbf{x}, \mathbf{y}) \lar \mathcal{S}^{2n_2}_{w, 1}(\Fqm)$ \\[0.25\baselineskip]
Compute $\mathbf{s} = \mathbf{x + h \cdot y} \mod P$ \\[0.25\baselineskip]
Output $\pk = \left( \mathbf{g}, \mathbf{h}, \mathbf{s} \right)$ 
and $\sk = \left( \mathbf{x}, \mathbf{y} \right)$

\vspace{\baselineskip}
\underline{\Encrypt$(\pk, \mathbf{m}, \theta)$}: \\[0.25\baselineskip]
Compute $\overline{\mathbf{g}}=(\mathbf{g} \, | \, 0 \hdots 0) \in \Fqm^{n_1 n_2}$ \\[0.25\baselineskip]
Compute the generator matrix $\mathbf{G} \in \mathbb{F}_{q^m}^{k\times (n_1 n_2)}$ of $\mathcal{G}^+_{\overline{\gv}}(n_1 n_2, n', k, m)$ \\[0.25\baselineskip]
Sample $(\mathbf{R}_1, \mathbf{E}, \mathbf{R}_2) \lar \mathcal{S}^{n_2 \times 3 n_1}_{w_1, w_2}(\Fqm)$ using randomness $\theta$ \\[0.25\baselineskip]
Compute $\mathbf{U} = \mathbf{R}_1 + \mathbf{h} \cdot \mathbf{R}_2$ and 
$\mathbf{V} = \fold(\mathbf{m}\mathbf{G}) + \mathbf{s} \cdot \mathbf{R}_2 + \mathbf{E}$ \\[0.25\baselineskip]
Output $\mathbf{C} = \left( \mathbf{U}, \mathbf{V} \right)$

\vspace{\baselineskip}
\underline{\Decrypt$(\pk, \sk, \mathbf{C})$}: \\[0.25\baselineskip]
Output $\mathbf{m} = \mathcal{G}^+_{\overline{\gv}}.\mathsf{Decode}\left( \unfold \left( \mathbf{V} - \mathbf{y}\cdot\mathbf{U} \right) \right)$

\end{minipage}
}
\caption{\label{fig:scheme-one} \schemeone encryption scheme}
\end{figure}

\vspace{-1.5\baselineskip}

\begin{theorem}\label{thm:scheme-one}
  The \schemeone scheme depicted in Figure~\ref{fig:scheme-one} is $\INDCPA$ under the $\DIRSD$ and 
  the $\DNHIRSL$ assumptions. 
\end{theorem}

\begin{proof}\label{proof-scheme-one}
  The proof of the \schemeone scheme is similar to the proof from \cite{AABBBBCDGHZ20} with an $\IRSD(m, 2n_2, n_2, \omega)$ instance defined from a $[2n_2, n_2]$ code and an $\NHIRSL(m, n_2, n_2, \omega_1, \omega_2, n_1)$ instance defined from a $[3n_2, n_2]$ code. These instances are defined by the following products:
\end{proof}

\vspace{-1\baselineskip}

$$
\left(\begin{matrix}
  \mathbf{I}_{n_2} & ~ \IM{\mathbf{h}} \\
\end{matrix}\right)
\times
\left(\mathbf{x}, \mathbf{y}\right)^\intercal
=
\mathbf{s}^\intercal,
$$
$$
\left(\begin{matrix}
  \mathbf{I}_{n_2} & ~ \mathbf{0}       & ~ \IM{\mathbf{h}} \\
  \mathbf{0}       & ~ \mathbf{I}_{n_2} & ~ \IM{\mathbf{s}}
\end{matrix}\right)
\times
\left(\mathbf{R}_1, \mathbf{E}, \mathbf{R}_2\right)^\intercal
=
\left(\mathbf{U},\mathbf{V} - \operatorname{Fold}(\mathbf{m} \mathbf{G})\right).
$$

\subsection{\schemetwo scheme}

Our new encryption scheme denoted \schemetwo stands for 
Multiple syndromes Unstructured Rank with Augmented Gabidulin codes encryption scheme.
It is particularly interesting security wise as it does not use structured codes 
contrarily to existing constructions such as ROLLO, RQC or our new proposal \schemeone. 
Indeed, it only relies on the security of the $\RSL$ problem.
\schemetwo leverages multiple syndromes and augmented Gabidulin codes.
In addition, it features two variants as it can be instantiated with either 
homogeneous or non-homogeneous errors.

\vspace{\baselineskip}
\noindent \textbf{Notations.}
Hereafter, $\fold$ and $\unfold$ refer to the procedure introduced in Section~\ref{sec:scheme-multi}. 
Let $\mathcal{S}_{w,1}^{n \times 2n_1}(\Fqm)$ and $\mathcal{S}_{(w_1,w_2)}^{n_2 \times (n + n_1 + n)}(\Fqm)$ be defined as:
\begin{align*}
  \mathcal{S}_{w,1}^{n \times 2n_1}(\Fqm) = & \ \{ \mathbf{X} = (\mathbf{X}_1, \mathbf{X}_2) \in \Fqm^{n \times 2n_1} \,\big|\, \norme{\mathbf{X}} = w, \ 1 \in  \Supp(\mathbf{X})\}, \\
  \mathcal{S}_{(w_1,w_2)}^{n_2 \times (n + n_1 + n)}(\Fqm) = & \ \{ 
      \mathbf{X} = (\mathbf{X}_1,\mathbf{X}_2,\mathbf{X}_3) \in \Fqm^{n_2 \times (n + n_1 + n)} \,\big|\,
      \norme{(\mathbf{X}_1,\mathbf{X}_3)} = w_1, \\
      & \norme{\mathbf{X}_2} = w_1+w_2, \quad \Supp(\mathbf{X}_1,\mathbf{X}_3) \subset \Supp(\mathbf{X}_2)\}.
\end{align*}

\vspace{\baselineskip}
\noindent \textbf{Protocol.}
The \schemetwo is described on Figure~\ref{fig:scheme-two}. 
It relies on two codes namely an augmented Gabidulin code 
$\mathcal{G}^+_{\overline{\gv}}(n, n', k, m)$ that can can correct up to 
$\delta := \left\lfloor \frac{n'-k+\varepsilon}{2} \right\rfloor$
errors using the efficient decoding algorithm 
$\mathcal{G}^+_{\overline{\gv}}.\mathsf{Decode}(.)$ as well as
a random $[2n, n]_{\fqm}$-code with 
parity check matrix $(\mathbf{I} ~ \mathbf{H})$.
The correctness of the protocol follows from: 
\begin{align*}
\mathbf{V} - \mathbf{U}\mathbf{Y} 
%&= \fold(\mathbf{m}\mathbf{G}) + \mathbf{R}_2 \mathbf{S} + \mathbf{E} - (\mathbf{R}_1 + \mathbf{R}_2 \mathbf{H}) \mathbf{Y} \\
&= \fold(\mathbf{m}\mathbf{G}) + \mathbf{R}_2 (\mathbf{X} + \mathbf{H} \mathbf{Y}) + \mathbf{E} - (\mathbf{R}_1 + \mathbf{R}_2 \mathbf{H})\mathbf{Y} \\
&= \fold(\mathbf{m}\mathbf{G}) + \mathbf{R}_2 \mathbf{X} - \mathbf{R}_1 \mathbf{Y} + \mathbf{E}.
\end{align*}

\noindent As a consequence, $\unfold\left(\mathbf{V}-\mathbf{Y}\mathbf{U}\right)=
\mathbf{m}\mathbf{G}+
\unfold\left(
\mathbf{X} \mathbf{R}_2
- \mathbf{Y} \mathbf{R}_1
+ \mathbf{E} \right) \in \Fqm^n
$
which means that $
\mathcal{G}_{\gv}.\mathsf{Decode}\left(
\unfold\left(
\mathbf{V}-
\mathbf{Y}\mathbf{U}
\right)\right)=\mathbf{m}
$
as long as: \[\norme{ \,
\unfold\left(
\mathbf{X}\mathbf{R}_2
- \mathbf{Y}\mathbf{R}_1
+ \mathbf{E} \right)
} \le \delta.\]

\begin{figure}[!ht]
\centering
\fbox{
\begin{minipage}{.965\textwidth}

\underline{\Setup$(1^\lambda)$} \\[0.25\baselineskip]
Generate and output \param{} = $(n, n', n_1, n_2, k, \epsilon, \delta, w, w_1,w_2)$ where $n = n_1 n_2$.

\vspace{\baselineskip}
\underline{\KeyGen$(\param)$}: \\[0.25\baselineskip]
Sample $\mathbf{g} \lar \mathcal{S}^{n'}_{n'}(\Fqm)$, $\mathbf{H} \lar \Fqm^{n \times n}$ and 
  $(\mathbf{X}, \mathbf{Y}) \lar \mathcal{S}^{n \times 2n_1}_{w, 1}(\Fqm)$ \\[0.25\baselineskip]
Compute $\mathbf{S} = \mathbf{X + HY}$ \\[0.25\baselineskip]
Output $\pk = \left( \mathbf{g}, \mathbf{H}, \mathbf{S} \right)$ 
and $\sk = \left( \mathbf{X}, \mathbf{Y} \right)$

\vspace{\baselineskip}
\underline{\Encrypt$(\pk, \mathbf{m}, \theta)$}: \\[0.25\baselineskip]
Compute $\overline{\mathbf{g}}=(\mathbf{g} \, | \, 0 \hdots 0) \in \Fqm^{n}$ \\[0.25\baselineskip]
Compute the generator matrix $\mathbf{G} \in \mathbb{F}_{q^m}^{k\times n}$ of $\mathcal{G}^+_{\overline{\gv}}(n, n', k, m)$ \\[0.25\baselineskip]
Sample $(\mathbf{R}_1, \mathbf{E}, \mathbf{R}_2) \lar \mathcal{S}^{n_2 \times (n + n_1 + n)}_{w_1, w_2}(\Fqm)$ using randomness $\theta$ \\[0.25\baselineskip]
Compute $\mathbf{U} = \mathbf{R}_1 + \mathbf{R}_2 \mathbf{H}$ and $\mathbf{V} = \fold(\mathbf{m}\mathbf{G}) + \mathbf{R}_2 \mathbf{S} + \mathbf{E}$ \\[0.25\baselineskip]
Output $\mathbf{C} = \left( \mathbf{U}, \mathbf{V} \right)$

\vspace{\baselineskip}
\underline{\Decrypt$(\pk, \sk, \mathbf{C})$}: \\[0.25\baselineskip]
Output $\mathbf{m} = \mathcal{G}^+_{\overline{\gv}}.\mathsf{Decode}\left( \unfold \left( \mathbf{V} - \mathbf{UY} \right) \right)$

\end{minipage}
}
  \caption{\label{fig:scheme-two} \schemetwo (with non-homogeneous errors) encryption scheme}
\end{figure}

\begin{theorem}\label{thm:scheme-two}
  The \schemetwo scheme is $\INDCPA$ 
  under the $\DRSL$ and the $\DNHRSL$ assumptions. 
\end{theorem}

\begin{proof}\label{proof-scheme-two}
  The proof of the \schemetwo scheme is similar to the proof from \cite{AABBBBCDGHZ20} with an $\RSL(m, 2n, n, \omega, n_1)$ instance defined from a $[2n, n]$ code and an $\NHRSL(m, n, \allowbreak n_1, \omega_1, \omega_2, n_2)$ instance defined from a $[2n + n_1, n]$ code. These instances are defined by the following products:
\end{proof}

\vspace{-0.5\baselineskip}
$$
\left(\begin{matrix}
  \mathbf{I}_{n} & ~ \mathbf{H} \\
\end{matrix}\right)
\times
\left(\mathbf{X}, \mathbf{Y}\right)^\intercal
=
\mathbf{S},
$$

% $$
% \left(\mathbf{R}_1, \mathbf{E}, \mathbf{R}_2\right)
% \times
% \left(\begin{matrix}
%   \mathbf{I}_{n_2} & ~ \mathbf{0}       & ~ \mathbf{H} \\
%   \mathbf{0}       & ~ \mathbf{I}_{n_2} & ~ \mathbf{S}
% \end{matrix}\right)^\intercal
% =
% \left(\mathbf{U},\mathbf{V} - \operatorname{Fold}(\mathbf{m} \mathbf{G})\right)^\intercal.
% $$

$$
\left(\mathbf{R}_1, \mathbf{E}, \mathbf{R}_2\right)
\times
\left(\begin{matrix}
  \mathbf{I}_{n} & ~ \mathbf{0}       & ~ \mathbf{H} \\
  \mathbf{0}       & ~ \mathbf{I}_{n_1} & ~ \mathbf{S}
\end{matrix}\right)^\intercal
=
\left(\mathbf{U},\mathbf{V} - \operatorname{Fold}(\mathbf{m} \mathbf{G})\right).
$$

\section{Security analysis}\label{sec:security_analysis}

  In this section, we provide the complexity to solve some hard problems 
  in rank-based cryptography.

  \subsection{Attacks on the \RSD problem \cite{GRS16,AGHT18,BBCGPSTV20}}\label{sec:attaques_rsd}

    There are two general classes of attacks to \RSD, based on combinatorial or algebraic techniques. On the one hand, combinatorial attacks can be seen as the equivalent of ISD-type attacks in the rank metric setting. Relying on \cite{GRS16,AGHT18}, we estimate that the complexity of the best combinatorial attack is in

    \begin{equation}\label{eq:cout_combi}
    \text{min}\left( 2^{(w-1)\left\lfloor \frac{(k+1)m}{n} \right\rfloor}, 
    2^{w\left\lceil \frac{(k+1)m}{n} \right\rceil -m} \right)
    \end{equation} 
    $\ff{q}$-operations. On the other hand, algebraic attacks on the \RSD problem are by modeling the decoding instance into a system of polynomial equations, and the overall cost is reduced to the one of solving this system. To design our parameters, we take into account the most recent algebraic attack, namely the MaxMinors attack \cite{BBCGPSTV20}. Its complexity in $\ff{q}$ operations is estimated to be
    \begin{equation}\label{eq:cout_mm}
    \mathcal{O}\left( q^{aw}
    \textstyle m\binom{n-k-1}{w}\binom{n-a}{w}^{\omega - 1}\right),
    \end{equation}
    where $a \geq 0$ the smallest integer such that $m\binom{n-k-1}{w} \geq \textstyle\binom{n-a}{w}-1$ and where $\omega$ is a linear algebra constant. 
    
  \subsection{Attacks on the \NHRSD problem}\label{sec:attaques_nhrsd}
    
    This section is dedicated to the first cryptanalysis of the \NHRSD problem by proposing two attacks which exploit the inhomogeneous structure of the error.
    
    \subsubsection{A new combinatorial attack.}
        In this section, we may assume for clarity that the $n_1$ leftmost coordinates of $\ev$ correspond to the part of weight $w_1 + w_2$, namely $\ev = (\ev_2,\ev_1, \ev_3)$, and we also adopt a systematic form for the parity-check matrix $\Hm_{\ev} = \begin{pmatrix} \Idm_{n+n_1-1} & * \end{pmatrix}$ of the public code $\mathcal{C}_{\ev} := \mathcal{C} \oplus \langle \ev \rangle$. The parity-check equations for this code which are traditionally used in this type of attack are as follows:
        
        \begin{enumerate}
        	\item those associated to the $n$ first rows of $\Hm_{\ev}$ provide $n$ linear relations over $\ff{q^m}$ which can be mapped into $nm$ relations over $\ff{q}$ between unknowns coming from $\ev_2$, $\ev_1$ and $\ev_3$.
        	\item those associated to the $n_1-1$ last rows of $\Hm_{\ev}$ give $(n_1-1)m$ equations over $\ff{q}$ in unknowns coming from the components of $\ev_1$ and $\ev_3$ only.
        \end{enumerate}
        
        Before describing our attack, let us recall how \cite{GRS16,AGHT18} would solve a non-structured $\RSD(m,2n,n,w_1)$ instance to recover $(\ev_1,\ev_3)$. The most enhanced version of \cite{AGHT18} consists in guessing a subspace $V$ of dimension $r_1 \geq w_1$ such that $\alpha S_1 \subset V$ for some element $\alpha \in \ff{q^m}^{*}$ instead of simply $S_1 \subset V$ as it provides a better success probability. Then, it aims at solving the linear system given by the parity-check equations from $2.$ The largest value of $r_1$ for which one may expect a unique solution is given by 
        \begin{equation}\label{eq:r1}
        r_1 := \textstyle\left\lfloor \frac{m(n-1)}{2n} \right\rfloor = m - \left\lceil \frac{m(n+1)}{2n} \right\rceil.
        \end{equation}
        The classical cost given in \cite{AGHT18} is then roughly
        \begin{equation}\label{eq:cout_2n}
        \textstyle \widetilde{\mathcal{O}}\left( q^{w_1(m-r_1) -m} \right) = \widetilde{\mathcal{O}}\left( q^{w_1\left\lceil \frac{m(n+1)}{2n} \right\rceil - m} \right),
        \end{equation}
        where $\widetilde{\mathcal{O}}$ hides a polynomial 
        factor corresponding to solving this linear system. To benefit from the inhomogeneous structure of $\ev$ from \NHRSD, our approach follows the natural path of making a guess on a random subspace $V$ of dimension $r \geq w_1$ such that $S_1 \subset V$ \emph{and} a random subspace $Z \subset \ff{q^m}/V$ 
        of dimension $\rho \in \{w_2..m-r\}$ such that $S_2 \subset V \oplus Z$.
        
        \begin{theorem}\label{theo:complex_combi}
        	Our proposed combinatorial algorithm runs in time
        	\begin{align}\label{eq:bigO_K}
        	\widetilde{\mathcal{O}}{\left(q^{(w_1+w_2)(m-r) - w_2\rho - m}\right)}.
        	\end{align}
        \end{theorem}
        
        The complexity given by Equation \eqref{eq:bigO_K} is of the same shape as Equation \eqref{eq:cout_2n} since the rest
        of our attack is totally similar to \cite{GRS16,AGHT18}: expressing the coordinates of $(\ev_1,\ev_3)$ in a fixed basis of $V$ yields $2nr$ variables over $\ff{q}$, while we get $n_1(r+\rho)$ variables over $\ff{q}$ by writing the coordinates of $\ev_2$ in a fixed basis of $V \oplus Z$. For the linear algebra step, $n_1(r+\rho)$ random equations from 1. are used in order to express all the variables from $\ev_2$ in terms of the other variables, and we are left with a linear system of $(n+n_1-1)m - n_1(r+\rho)$ equations over $\ff{q}$ in only $2nr$ variables. This leads to the condition
        \begin{equation*}\label{eq:condi_linear}
        2nr \leq m(n+n_1-1) - n_1(r+\rho)
        \end{equation*}
        in order to expect at most one solution. 
        Overall, the main task to prove Theorem \ref{theo:complex_combi} is to compute the success probability $\Pi :=  
        \probab{S_1 \subset V,~S_2 \subset V \oplus Z}{V,Z}$, see Appendix \ref{sec:appendix_Pi}. Using \cite{AGHT18}, recall also that one may take advantage of $\ff{q^m}$-linearity by considering a greater probability of the form 
        \begin{equation}
        \probab{\exists \alpha \in \ff{q^m}^{*},~\alpha S_1 \subset V ,~\alpha S_2 \subset V \oplus Z}{V,Z} \approx \frac{q^{m}-1}{q-1}\Pi,
        \end{equation}
        and in case of success decoding the word $\alpha\ev$ instead of $\ev$. For clarity Appendix \ref{sec:appendix_combi} presents the plain version of the attack, but as this trick is compatible with our analysis the corresponding $q^{-m}$ factor appears in Equation \eqref{eq:bigO_K}. Finally, one has to consider the couple $(r,\rho)$ which leads to the best exponent in Equation \eqref{eq:bigO_K}. In other words, the goal will be to maximize the quantity $(w_1+w_2)r + w_2\rho$ under the constraints
        $(2n + n_1)r + n_1\rho \leq m(n+n_1-1)$,
        $w_1 \leq r$,
        $w_2 \leq \rho$,
        $r + \rho \leq m-1$,
        where $r,\rho \in \mathbb{N}$.
        
        This is an example of integer linear program (ILP), 
        and to solve this instance we have used dedicated tools.

    \subsubsection{Adaptation of the algebraic attack of \cite{BBCGPSTV20} against \NHRSD.}

      A first approach of this attack was 
      proposed in \cite{AABBBBCDGHZ20}, we build upon 
      this work and give a thorough analysis of the 
      complexity of this attack. 
      
      \begin{theorem}\label{theo:cost_mm}
      	Let $a \geq 0$ the smallest integer such that $$\mathcal{N}_{\ff{q}} \geq \textstyle\binom{2n+n_1 -a}{w_1+w_2} - M_a - \mathcal{\nu}_{\ff{q}}-1,$$ where $\mathcal{N}_{\ff{q}} =   m\textstyle\sum_{i=w_2}^{w_1+w_2} \binom{n_1-1}{i}\binom{n}{w_1 + w_2 - i}$, $\mathcal{\nu}_{\ff{q}} =  \textstyle m\binom{n_1-1}{w_2 - 1}\binom{n-1}{w_1}$ and $\textstyle M_a := \sum_{i=0}^{\omega_2 - 1} \binom{n_1}{i}\binom{2n-a}{\omega_1 + \omega_2 - i}$. The hybrid MaxMinors attack adapted to \NHRSD costs
      	\begin{equation*}
      	\mathcal{O}\left( q^{aw_1}\mathcal{N}_{\ff{q}}
      	\textstyle\left( \binom{2n+n_1-a}{w_1+w_2} - M_a - \mathcal{\nu}_{\ff{q}} \right)^{\omega - 1}\right)
      	\end{equation*}
      	operations in $\ff{q}$, where $\omega$ is a linear algebra constant.
      \end{theorem}
      
      \paragraph{\textbf{MaxMinors linear system \cite{BBCGPSTV20}.}}
        The MaxMinors system is a system of equations over $\ff{q^m}$ which vanish on the solutions to the \RSD instance. 
        Let $\yv = \cv + \ev \in \ff{q^m}^{(2n+n_1)}$ be the noisy codeword to be decoded in a random $\ff{q^m}$-linear code $\mathcal C$ of length $2n+n_1$ and dimension $n$ with generator matrix $\Gm \in \ff{q^m}^{n \times (2n+n_1)}$. The
        extended code
        $\mathcal C_{\ev} = \mathcal C_{\yv}$ is generated by the matrix $\textstyle \Gm_{\yv} :=
        \begin{pmatrix}
        \Gm\\\yv
        \end{pmatrix}$, and we also consider $\Hm_{\yv}\in\mathbb F_{q^m}^{(n+n_1-1)\times (2n+n_1)}$ a full-rank parity-check matrix for this code. We clearly have
        \begin{equation*}\label{eq:start_mm_sm}
        0 = \ev\trsp{\Hm_{\yv}} = \mat \beta \Mat(\ev)\trsp{\Hm_{\yv}} = \mat\beta\Sm\Cm\trsp{\Hm_{\yv}},
        \end{equation*}
        so that the matrix $\Cm\trsp{\Hm_{\yv}}$ contains a non-zero vector $\mat\beta\Sm$ in its left kernel and cannot be full-rank. In particular, the MaxMinors system is the system of maximal minors $\mathcal{P} := \left\lbrace P_J \right\rbrace_{J}$ such that $P_J := \left\vert \Cm \trsp{\Hm_{\yv}} \right\vert_{*,J}$ for each subset $J \subset \{1..n+n_1-1\}$, $\# J =w_1+ w_2$. The crux is that these equations are actually \emph{linear} in the minor variables $c_T : =\left\vert\Cm \right\vert_{*,T} \in \ff{q}$ by using the Cauchy-Binet formula for the determinant of a product of rectangular matrices, see \cite{BBBGNRT20,BBCGPSTV20}. In this section, the $c_T$'s will be sorted with respect to the following ordering on the $T$'s: we consider that
        $T=\lbrace t_1<\dots<t_r\rbrace < T'=\lbrace t'_1<\dots<t'_r\rbrace$
        if $t_j=t'_j$ for $j<j_0$ and $t_{j_0}<t'_{j_0}$ assuming that
        $1<2<\dots<n$. We will further assume that $\Hm_{\yv} :=
        \begin{pmatrix}
        * & \Idm_{n+n_1-1}
        \end{pmatrix}$ and from that assumption \cite{BBCGPSTV20} derive the fundamental Lemma \ref{lem:P_Jlem} on the shape of the MaxMinors equations:

        \begin{lemma}[Prop. 2, \cite{BBCGPSTV20}]\label{lem:P_Jlem}
        	\begin{eqnarray}\label{eq:P_J}
        	P_J      = c_{J+n+1} + \sum_{\substack{T^-\subset\lbrace 1..n+1\rbrace, T^+\subset (J+n+1)\\T = T^-\cup T^+,~\#T=w_1 + w_2,~T^-\ne\emptyset}} c_T \vert \Hm_{\yv} \vert_{J,T}.
        	\end{eqnarray}
        \end{lemma}

        A direct consequence of Lemma \ref{lem:P_Jlem} is that the equations of $\mathcal{P}$ are linearly independent over $\ff{q^m}$ as their leading terms are distinct. 

      \paragraph{\textbf{Removing variables corresponding to zero minors.}}

        The very same MaxMinors system can be employed to attack \NHRSD. A main difference in this case is that if one wants to decrease the number of minor variables by relying on the special structure of $\ev$ as shown in \cite{AABBBBCDGHZ20,BBCGPSTV20}, then linear relations between the equations after removing these variables also occur and must be taken into account in the analysis. 
        Recall from \cite[6.2.2]{AABBBBCDGHZ20} that the row 
        support of $\Mat(\ev)\in \ff{q}^{m \times (n+n_1+n)}$ can be written as
        \begin{equation}\label{eq:C_rqc}
        \Cm = \begin{pmatrix}
        \Cm_1 & \Cm_2 & \Cm_3 \\
        0 & \Cm'_2 & 0
        \end{pmatrix} \in \ff{q}^{(w_1 + w_2) \times (n+n_1+n)},
        \end{equation}
        where $\Cm_1,~\Cm_3 \in \ff{q}^{w_1 \times n}$, $\Cm_2 \in \ff{q}^{w_1 \times n_1}$ and $\Cm'_2\in \ff{q}^{w_2 \times n_1}$. From Equation \eqref{eq:C_rqc}, it has been noted that the minors $\left\vert \Cm\right\vert_{*,T}$ such that $T \cap \{n+1..n+n_1\} \leq w_2 - 1$ are always zero. This means that the  
        \begin{equation}\label{eq:removed_M}
        M := \sum_{i=0}^{w_2 - 1} \binom{n_1}{i}\binom{2n}{w_1 + w_2 - i}
        \end{equation}
        variables from the set
        \begin{equation*}
        \zeta := \left\lbrace c_T,~T \subset \{1..(2n+n_1)\},~\# T = w_1 + w_2,~T \cap \{n+1..n+n_1\} \leq w_2 - 1 \right\rbrace
        \end{equation*}
        can be set to zero in the MaxMinors system. It is then relevant to separate the initial $P_J$ equations into several subsets in function of the presence or the absence of these $c_T$ variables. We consider the partition $\mathcal P := \mathcal P_{\text{lost}} \sqcup \mathcal P_{\text{rest}} \sqcup \mathcal P_{\text{indep}}$, where 
        \begin{align*}
        \mathcal P_{\text{lost}} & :=   \left\lbrace P_J : ~\#J=w_1+w_2,~ \#(J\cap\lbrace 1..(n_1-1)\rbrace)\le w_2 - 2 \right\rbrace \\
        \mathcal P_{\text{rest}} & :=   \left\lbrace P_J : ~\#J=w_1+w_2,~ \#(J\cap\lbrace 1..(n_1-1)\rbrace) = w_2-1 \right\rbrace \\
        \mathcal P_{\text{indep}} & :=   \left\lbrace P_J : ~\#J=w_1+w_2,~ \#(J\cap\lbrace 1..(n_1-1)\rbrace)\ge w_2 \right\rbrace.
        \end{align*}

        Using Lemma \ref{lem:P_Jlem}, it is easy to grasp the shape of the equations from $\mathcal P_{\text{lost}}$ and $\mathcal P_{\text{indep}}$ after removing the minor variables belonging to $\zeta$:

        \begin{proposition}\label{prop:MM_easy}
        	After setting the minor variables from $\zeta$ to zero in the MaxMinors system $\mathcal{P}$, we have the following properties:
        	\begin{enumerate}
        		\item \label{case:Q} The equations in 
        		$\mathcal P_{\text{lost}}$ all become zero. 
        		\item \label{case:Qk+1} The equations in $\mathcal P_{\text{indep}}$ keep the same leading terms and therefore they are still linearly independent. We have
        		\begin{equation*}
        		\dim_{\ff{q^m}}\left \langle \mathcal P_{\text{indep}} \right\rangle = \#\mathcal P_{\text{indep}} = \textstyle\sum_{i=w_2}^{w_1+w_2} \binom{n_1-1}{j}\binom{n}{w_1 + w_2 - j}.
        		\end{equation*} 
        		Finally, the system $\mathcal P_{\text{indep}}$ contains at most $\textstyle\binom{2n+n_1}{w_1+w_2} - M$ variables.
        	\end{enumerate}
        \end{proposition}
        \begin{proof}
        	See Appendix \ref{sec:proof_prop1}. \qed
        \end{proof}

        Contrary to $\mathcal{P}_{\text{indep}}$, the equations in $\mathcal P_{\text{rest}}$ have their leading terms in $\zeta$ so that these monomials are destroyed after setting the $M$ minor variables to zero. More precisely, by Lemma \ref{lem:P_Jlem}, an equation $P_J \in \mathcal P_{\text{rest}}$ becomes
        \begin{eqnarray}\label{eq:P_J_tilde}
        \widetilde{P_J}      & = \sum_{\substack{T^-\subset\lbrace 1..n+1\rbrace,~T^+\subset (J+n+1)\\T = T^-\cup T^+,~n+1 \in T^{-},~\#(T^{+} \cap \{n+2..n+n_1\}) = w_2 - 1}} c_T \vert \Hm_{\yv} \vert_{J,T} \notag \\
        & = \sum_{\substack{T^-\subset\lbrace 1..n+1\rbrace,~T^+\subset (J+n+1)\\T = T^-\cup T^+,~n+1 \in T^{-},~T^{+} \cap \{n+2..n+n_1\} = (J \cap \{1..(n_1-1)\}) + n+1}} c_T \vert \Hm_{\yv} \vert_{J,T}.
        \end{eqnarray}	

        For clarity, we still denote the resulting system by $\mathcal P_{\text{rest}}$. We analyze it in the following Proposition \ref{prop:MM}:

        \begin{proposition}\label{prop:MM}
        	After setting the minor variables from $\zeta$ to zero in $\mathcal P_{\text{rest}}$, one obtains a system of rank
        	$\textstyle \binom{n_1-1}{w_2 - 1}\binom{n-1}{w_1}$ and whose equations are also independent from $\mathcal P_{\text{indep}}$. Finally, these equations contain at most $\textstyle\binom{n_1-1}{w_2-1}\binom{2n}{w_1}$ variables.
        \end{proposition}

        The first part of Proposition \ref{prop:MM} is obvious. Using Equation \eqref{eq:P_J_tilde}, the leading term of $\widetilde{P_J} \in \mathcal{P}_{\text{rest}}$ is a $c_T$ variable such that $n+1 \in T$, whereas the leading term of any $P_{J'} \in \mathcal{P}_{\text{indep}}$ is $c_{J'+n+1}$ and $c_{J'+n+1} > c_T$ for any such $T$. Thus, what is left to prove in Proposition \ref{prop:MM} is that 
        $\dim_{\ff{q^m}}\left \langle \mathcal P_{\text{rest}} \right\rangle = \textstyle \binom{n_1-1}{w_2 - 1}\binom{n-1}{w_1}$ and that the number of variables is $\textstyle\binom{n_1-1}{w_2-1}\binom{2n}{w_1}$. For this we rely on the following lemma, whose proofs can be found in Appendix \ref{sec:proof_prop2}:

        \begin{lemma}\label{lem:P_resta}
        	For $A \subset \{n+2..n+n_1\},~\#A = w_2-1$, let
        	\begin{equation*}\label{eq:P_resta}
        	\mathcal P_{\text{rest},A} :=   \left\lbrace P_J \in \mathcal P_{\text{rest}}  : J\cap\lbrace 1..n_1-1\rbrace = A - (n+1) \right\rbrace,
        	\end{equation*}
        	so that $\left\lbrace \mathcal{P}_{\text{rest},A}\right\rbrace_{A}$ is a partition of $\mathcal{P}_{\text{rest}}$. We have $\left \langle \mathcal P_{\text{rest}} \right\rangle = \displaystyle \oplus_{A} \left \langle \mathcal P_{\text{rest},A} \right\rangle$. 
        \end{lemma}

        \begin{lemma}\label{lem:dim_P_resta}
        	For $A \subset \{n+2..n+n_1\},~\#A = w_2-1$, let $\mathcal P_{\text{rest},A}$ as defined in Lemma \ref{lem:P_resta}. With very high probability, we have $\dim_{\ff{q^m}}\left \langle \mathcal P_{\text{rest},A} \right\rangle = \textstyle \binom{n-1}{w_1}$.
        \end{lemma}

      \paragraph{\textbf{Finishing the attack by projecting over $\ff{q}$.}}
        The last step of the initial MaxMinors attack on \RSD is by solving the ``projected" linear system $\mathcal{P}_{\ff{q}} := \left\lbrace P_{j,J} \right\rbrace_{j,J}$ obtained by expressing the coefficients of the $P_J$'s in a fixed basis of $\ff{q^m}$ over $\ff{q}$ and taking each component, yielding $m$ times more equations. We proceed in a very similar way as in \cite{BBCGPSTV20} and due to space constraints we do not recall all the details of this step. Our final complexity estimate relies on 
    
        \begin{ass}\label{cor:coro_PJ}
        	Let $\mathcal{P}_{\text{indep},\ff{q}}$ (resp. $\mathcal{P}_{\text{rest}, \ff{q}}$) be the system over $\ff{q}$ obtained by projecting $\mathcal{P}_{\text{indep}}$ (resp. $\mathcal{P}_{\text{rest}}$) where the variables in $\zeta$ had already been removed, let $\mathcal{N}_{\ff{q}} := \dim_{\ff{q}}\left \langle \mathcal{P}_{\text{indep},\ff{q}} \right\rangle$, let $\mathcal{\nu}_{\ff{q}} := \dim_{\ff{q}}\left \langle \mathcal{P}_{\text{rest},\ff{q}}\right\rangle$ and let $M$ as defined in Equation \eqref{eq:removed_M}. We assume that 
        	\begin{align}\label{eq:N_fq}
        	\mathcal{N}_{\ff{q}} & =  m\dim_{\ff{q^m}}\left \langle \mathcal{P}_{\text{indep}} \right\rangle = m\textstyle\sum_{i=w_2}^{w_1+w_2} \binom{n_1-1}{i}\binom{n}{w_1 + w_2 - i}
        	\end{align}
        	when this value is $\leq \textstyle\binom{2n+n_1}{w_1+w_2}-M$ and $\mathcal{N}_{\ff{q}} = \textstyle\binom{2n+n_1}{w_1+w_2} - M-1$ otherwise, and 
        	\begin{align}\label{eq:nu_fq}
        	\mathcal{\nu}_{\ff{q}} & = m\dim_{\ff{q^m}}\left \langle  \mathcal{P}_{\text{rest}} \right\rangle =  \textstyle m\binom{n_1-1}{w_2 - 1}\binom{n-1}{w_1},
        	\end{align}
        	provided that this value is $\leq \textstyle\binom{n_1-1}{w_2-1}\binom{2n}{w_1}$.
        \end{ass}

          To solve the final system, one can start by performing linear algebra on $\mathcal{P}_{\text{rest},\ff{q}}$ and then substitute $\mathcal{\nu}_{\ff{q}}$ variables corresponding to an echelonized basis of 
            $\left \langle \mathcal{P}_{\text{rest},\ff{q}} \right \rangle$ in the system $\mathcal{P}_{\text{indep},\ff{q}}$ to get a new system $\mathcal{P}'_{\text{indep},\ff{q}}$. The final step is then to solve the linear system $\mathcal{P}'_{\text{indep},\ff{q}}$ in $\textstyle\binom{2n+n_1}{w_1+w_2} - M - \mathcal{\nu}_{\ff{q}}$ variables.
            
            \begin{coro}[Same notations as in Assumption \ref{cor:coro_PJ}]\label{theo:final_cost}
            	\ Let $\mathcal{P}_{\text{indep},\ff{q}}$ and let \ $\mathcal{P}_{\text{rest},\ff{q}}$ denote the projected systems from Assumption \ref{cor:coro_PJ}. We consider $\mathcal{P}'_{\text{indep},\ff{q}}$ the linear system obtained from $\mathcal{P}_{\text{indep},\ff{q}}$ after plugging $\mathcal{\nu}_{\ff{q}}$ equations from the echelon form of $\mathcal{P}_{\text{rest},\ff{q}}$ to substitute variables. Assuming that the system $\mathcal{P}'_{\text{indep},\ff{q}}$ can be solved, namely $\mathcal{N}_{\ff{q}} \geq \textstyle\binom{2n+n_1}{w_1+w_2} - M - \mathcal{\nu}_{\ff{q}}-1$, the complexity of solving the system is 
            	\begin{equation*}\label{eq:final_complex}
            	\mathcal{O}\left( \mathcal{N}_{\ff{q}}
            	\textstyle\left( \binom{2n+n_1}{w_1+w_2} - M - \mathcal{\nu}_{\ff{q}} \right)^{\omega - 1}\right)
            	\end{equation*}
            	operations in $\ff{q}$, where $\omega$ is a linear algebra constant.
            \end{coro}

            However, the projected linear system cannot be solved directly when there are not enough equations compared to the number of minor variables, \emph{i.e.} $\mathcal{N}_{\ff{q}} < \textstyle\binom{2n+n_1}{w_1+w_2} - M - \mathcal{\nu}_{\ff{q}}-1$. In this case, a method suggested in \cite{BBCGPSTV20} is an hybrid approach by adding linear constraints on these minor variables which are obtained by fixing the entries of $a \geq 0$ columns in the matrix $\Cm$. Here, like it was done in \cite{AABBBBCDGHZ20}, 
            it is possible to take advantage of the particular 
            structure of $\Cm$ given in Equation \eqref{eq:C_rqc} by fixing columns containing 
            only $w_1$ non-zero coordinates, which leads to a smaller exponential factor of $q^{aw_1}$ in the final cost instead of the naive $q^{a(w_1+w_2)}$. The cost claimed in Theorem \ref{theo:cost_mm} follows.  
        
  \subsection{Attacks on the \RSL problem}\label{sec:attaques_rsl}

      In this section, we consider an \RSL$(m,n,k,r,N)$-instance, say $N$ distinct 
      \RSD instances whose errors share the same support of dimension $r$. 
      This number $N$ is a crucial parameter to estimate the hardness of \RSL 
      and in particular to compare it to \RSD. For instance, this problem can be solved in 
      polynomial time when $N \geq nr$ due to \cite{GHPT17a}. 
      A more powerful attack was later found in \cite{DT18b} and it suggests that 
      secure \RSL instances must satisfy a stronger condition: $N < kr$. 

      In what follows, we give a new combinatorial attack against \RSL, it is more 
      efficient than the previous combinatorial attacks, plus it enables us to decrease 
      the threshold where the \RSL problem starts to be solvable in polynomial time. 
      In addition to this, we give more explicit formulas 
      to clarify the recent algebraic attack of \cite{BB21}. 

    \subsubsection{New combinatorial attack on \RSL.}

      \begin{theorem}[Combinatorial attack on \RSL]\label{thm:RSL_combi}
        There exists a combinatorial attack on \RSL$(m,n,k,r,N)$ with complexity 
        \[
        \widetilde{\mathcal{O}}\left(q^{r\left(m-\left\lfloor \frac{m(n-k)-N}{n-a} \right\rfloor\right)}\right)
        \]
        operations in $\fq$, where $a:=\left\lfloor \frac{N}{r} \right\rfloor$.
      \end{theorem}
      \begin{proof}
     
      	Let $\sv_i \in \ff{q^m}^{n-k},~1 \leq i \leq N$ denote the $N$ syndromes from the \RSL instance. By definition there exist $\ev_i \in \ff{q^m}^n,~\norme{\ev_i}=r,~\Hm \trsp{\ev_i} =\trsp{\sv_i}$, where $\Hm\in \fqm^{(n-k) \times n}$ is a parity-check matrix and where $\text{Supp}(\ev_i)$ does not depend on $i$. Similarly to \cite{GRSZ14,BB21}, this last property enables us to use the fact that there exists an $\mathbb{F}_q$-linear combination $(\ \boldsymbol{0}_a\ |\ \widetilde{\ev} \ ) \in \ff{q^m}^n$ of the $\ev_i$'s which is all-zero on its first $a:=\left\lfloor \frac{N}{r} \right\rfloor$ coordinates. This error corresponds to a \emph{secret} linear combination of the syndromes, more precisely
      	\begin{equation*}
        \textstyle \exists \lambda_1,\lambda_2,\hdots,\lambda_N \in \mathbb{F}_q, \quad 
        \Hm \trsp{(\ \boldsymbol{0}_a\ |\ \widetilde{\ev} \ )} = \sum_{i=1}^{N} \lambda_i \trsp{\sv_i}. 
      	\end{equation*}
      	By setting $\widetilde{\Hm} := \Hm_{*,\{a+1\dots n\}}$, this is equivalent to
      	
      	\begin{equation}\label{eq:RSL_combi}
      	\widetilde{\Hm}\trsp{\widetilde{\ev}} = \sum_{i=1}^{N} \lambda_i \trsp{\sv_i}.
      	\end{equation}
        Equation \eqref{eq:RSL_combi} can be seen as $n-k$ parity-check equations which may be exploited by the classical combinatorial technique, see \cite{GRS16,AGHT18} or the discussion above Equation \eqref{eq:r1}. The main difference here is that the right hand this equation also contains $N$ unknowns $\lambda_i \in \fq$. Still, we can pick a vector space 
        $V$ of dimension $r_1 \ge r$ and hoping that $\text{Supp}\left(\widetilde{\ev} \right) \subset V$. If this is the case, one can derive from \eqref{eq:RSL_combi} a linear system of $(n-k)m$ equations over $\ff{q}$ in $N+(n-a)r_1$ variables, where the first $N$ variables merely correspond to the $\lambda_i$'s. The final cost is then obtained by looking at the optimal value of $r_1$ which allows to solve this linear system, namely 
        $r_1 := \left\lfloor \frac{m(n-k)-N}{n-a} \right\rfloor.$ \qed
      \end{proof}

      Thanks to Theorem \ref{thm:RSL_combi}, we are able to derive a 
      value of $N$ so that an \RSL instance is solvable in polynomial time, 
      this is the topic of Corollary \ref{cor:bound_polynomial}. 

      \begin{corollary}[New Bound for \RSL]\label{cor:bound_polynomial}
        \ An \RSL instance with parameters $\allowbreak (m,n,k,r,N)$ can be solved in polynomial time 
        using the attack of Theorem \ref{thm:RSL_combi} as long as 
        \[
          N>kr\frac{m}{m-r}.
        \]
      \end{corollary}
      \begin{proof}
        This is a straightforward application of the complexity given by 
        Theorem \ref{thm:RSL_combi} 
        which states that the attack has a polynomial cost (hidden in the 
        $\widetilde{\mathcal{O}}$) and an exponential cost of 
        $q^{r(m-\delta)}$ where 
        $
          \delta = \left\lfloor \frac{m(n-k)-N}{n-a} \right\rfloor.
        $ 
        Since $r\ne 0$ by definition, the only way for the exponential component to vanish is 
        if $\delta=m$. Without loss of generality, we assume that $\frac{N}{r}$ and 
        $\frac{m(n-k)-N}{n-N/r}$ are integers. 
        By solving a simple equation, one gets that 
        $\delta=m \Longleftrightarrow N=kr\frac{m}{m-r}$, hence the result. \qed
      \end{proof}
      Note that, as long as $\frac{m}{m-r}$, which is often the case for cryptographic parameters, 
      our bound is lower than the previous one, given in \cite{GHPT17a}, 
      which was $N>nr$.

    \subsubsection{Algebraic attack of \cite{BB21}.}
      This attack consists in solving a bilinear system at some 
      bi-degree $(b,1)$ for $b \geq 1$ by using an XL approach 
      similar to \cite{BBCGPSTV20}. The two cases $``\delta = 0"$ 
      and $``\delta > 0"$ presented below correspond to two different 
      specializations of this bilinear system which lead to different 
      costs. Here, we provide explicit formulas  to compute these two 
      complexities (for the binary field $\ff{2}$). In particular, we 
      also include the values of $\alpha_R$ and $\alpha_{\lambda}$ which correspond 
      to the hybrid approach mentioned in \cite{BB21}. 
      Finally, note that these formulas are valid only when $N>n-k-r$.
      
      \paragraph{\textbf{First case: $\delta=0$.}} 
      Let $a$ be the unique integer such that $ar < N \le (a+1)r$, and let $N':=ar+1$.
      For $1 \le b \le r+1$, 
      the number of variables for linearization is 
      \begin{equation}\label{eq:mb}
      \mathcal{M}_{\le b}^{\mathbb{F}_2} := 
      \sum_{i=1}^{b} 
      \binom{n-a-\alpha_R}{r}
      \binom{N'-\alpha_\lambda}{i}, 
      \end{equation}
      where $0 \le \alpha_R <n-a-r$, and 
      $0 \le \alpha_\lambda<N'-b$, and the number of linearly independent equations at hand is equal to $m\mathcal{N}_{\le b}^{\mathbb{F}_2}$ where
      \begin{equation}\label{eq:nb}
      \mathcal{N}_{\le b}^{\mathbb{F}_2} := 
      \sum_{i=1}^{b} 
      \sum_{d=1}^{i}\sum_{j=1}^{n-k} 
      \binom{j-1}{d-1}\binom{n-k-j}{r-d+1}\binom{N'-\alpha_\lambda-j}{i-d}.
      \end{equation}
      The complexity is given by 
      \begin{multline}\label{complexite_RSL}
          \mathcal{O}
      \Bigg(\operatorname{min}\Bigg(
      2^{r\alpha_R+\alpha_\lambda}
      m\mathcal{N}_{\le b}^{\mathbb{F}_2}
      (\mathcal{M}_{\le b}^{\mathbb{F}_2})^{\omega-1}, \\ 
      2^{r\alpha_R+\alpha_\lambda}
      (N'-\alpha_\lambda)\binom{k-a+1+r}{r}(\mathcal{M}_{\le b}^{\mathbb{F}_2})^2
      \Bigg)\Bigg)
      \end{multline}
      provided that $m\mathcal{N}_{\le b}^{\mathbb{F}_2} \ge \mathcal{M}_{\le b}^{\mathbb{F}_2} -1$, and 
      where the values of $b,\alpha_R$, and $\alpha_\lambda$ are chosen to minimize the complexity.

      \paragraph{$\mathbf{\boldsymbol{\delta}>0}$ \textbf{case}.}
      Let $\delta$ be a positive integer such that $N \ge \delta(n-r+\delta)$, 
      let $a$ be the greatest integer such that 
      $N > \delta(n-r+\delta)+a(r-\delta)$ and let $N':=\delta(n-r+\delta)+a(r-\delta)$.
      To find the complexity of this attack, one replaces $r$ by $r-\delta$ in the expressions of  $\mathcal{M}_{\le b}^{\mathbb{F}_2}$ and 
      $\mathcal{N}_{\le b}^{\mathbb{F}_2}$ from Equations (\ref{eq:mb}) and 
      (\ref{eq:nb}). The complexity is finally obtained with Equation (\ref{complexite_RSL}) and its minimal value now depends on $\delta>0$ as well as $b,\alpha_R$, and $\alpha_\lambda$ as above.

    \subsubsection{Visualization of the attacks against \RSL.}

\begin{figure}[ht!]
  \centering
  \caption{Complexity $\mathcal{C}$ (in bits) of the best known attacks against an 
  RSL instance with parameters $[m,n,k,r]=[61, 100, 50, 7]$ in terms of the 
  number $N>n-k-r$ of syndromes. In the legend: $\mathbf{C}$ stands for 
  our combinatorial attack (see Theorem \ref{thm:RSL_combi}), 
  all the other symbols correspond 
  to the 2 cases of the algebraic attack \cite{BB21} where 
  the ``*'' indicates the use of Wiedemann algorithm instead of Strassen's. \label{fig:RSL}}
  \begin{tikzpicture} 
    \begin{axis}
      [
      %enlarge x limits=false, % permet d'arrêter le graphique à la valeur max, sans marge automatique
      %enlarge y limits=false, % permet d'arrêter le graphique à la valeur max, sans marge automatique
      legend entries={$\delta_{=0}$,$\delta_{=0}^*$,$\delta_{>0}$,\ $\delta_{>0}^*$ \ , $\mathbf{C}$},
      legend pos=south west,
      height=8cm,
      width=12.5cm,
      xmin=10,xmax=305,
      ymin=80,ymax=220,
      xlabel={Number of syndromes $N$},
      ylabel={$\mathcal{C}$ \hspace*{-5mm}},
      ylabel style={rotate=-90}]
        \addplot [
            scatter,
            only marks,
            point meta=explicit symbolic,
            scatter/classes={
                d_zero={mark=square*,blue},
                d_zero_w={mark=triangle*,blue},
                d_pos={mark=square*,red},
                d_pos_w={mark=triangle*,red},
                e={mark=square*,olive}    % <-- don’t add comma
            },
        ] table [meta=label] {
            x      y        label
            44    138   d_zero
            
            50    138   d_zero
            
            56    138   d_zero
            
            62    138   d_zero
            
            68    138   d_zero
            
            74    138   d_zero
            
            80    138   d_zero
            
            86    138   d_zero
            
            92    138   d_zero
            
            98    138   d_zero
            
            104   138   d_zero
            
            110   138   d_zero
            
            116   138   d_zero
            
            122   138   d_zero
            
            128   138   d_zero
            
            134   138   d_zero
            
            140   138   d_zero
            
            146   138   d_zero
            
            152   138   d_zero
            
            158   138   d_zero
            
            164   138   d_zero
            
            170   138   d_zero
            
            176   138   d_zero
            
            182   138   d_zero
            
            188   138   d_zero
            
            194   138   d_zero
            
            200   138   d_zero
            
            206   138   d_zero
            
            212   138   d_zero
            
            218   138   d_zero
            
            224   138   d_zero
            
            230   138   d_zero
            
            236   138   d_zero
            
            242   138   d_zero
            
            248   138   d_zero
            
            254   138   d_zero
            
            260   138   d_zero
            
            266   138   d_zero
            
            272   138   d_zero
            
            278   138   d_zero
            
            284   138   d_zero
            
            290   138   d_zero
            
            296   138   d_zero
            
            300   138   d_zero

            44      138   d_zero_w
            
            50      138   d_zero_w
            
            56      138   d_zero_w
            
            62      138   d_zero_w
            
            68      138   d_zero_w
            
            74      138   d_zero_w
            
            80      138   d_zero_w
            
            86      138   d_zero_w
            
            92      138   d_zero_w
            
            98      138   d_zero_w
            
            104     138   d_zero_w
            
            110     138   d_zero_w
            
            116     138   d_zero_w
            
            122     138   d_zero_w
            
            128     138   d_zero_w
            
            134     138   d_zero_w
            
            140     138   d_zero_w
            
            146     138   d_zero_w
            
            152     138   d_zero_w
            
            158     138   d_zero_w
            
            164     138   d_zero_w
            
            170     138   d_zero_w
            
            176     138   d_zero_w
            
            182     138   d_zero_w
            
            188     138   d_zero_w
            
            194     138   d_zero_w
            
            200     138   d_zero_w
            
            206     138   d_zero_w
            
            212     138   d_zero_w
            
            218     138   d_zero_w
            
            224     138   d_zero_w
            
            230     138   d_zero_w
            
            236     138   d_zero_w
            
            242     138   d_zero_w
            
            248     138   d_zero_w
            
            254     138   d_zero_w
            
            260     138   d_zero_w
            
            266     138   d_zero_w
            
            272     138   d_zero_w
            
            278     137   d_zero_w
            
            284     133   d_zero_w
            
            290     130   d_zero_w
            
            296     123   d_zero_w
            
            300     123   d_zero_w

            101   186   d_pos
            
            107   186   d_pos
            
            113   186   d_pos
            
            119   186   d_pos
            
            125   186   d_pos
            
            131   186   d_pos
            
            137   186   d_pos
            
            143   186   d_pos
            
            149   186   d_pos
            
            155   186   d_pos
            
            161   186   d_pos
            
            167   186   d_pos
            
            173   186   d_pos
            
            179   186   d_pos
            
            185   186   d_pos
            
            191   186   d_pos
            
            197   186   d_pos
            
            203   186   d_pos
            
            209   186   d_pos
            
            215   186   d_pos
            
            221   186   d_pos
            
            227   186   d_pos
            
            233   186   d_pos
            
            239   186   d_pos
            
            245   186   d_pos
            
            251   186   d_pos
            
            257   186   d_pos
            
            263   186   d_pos
            
            269   186   d_pos
            
            275   182   d_pos
            
            281   176   d_pos
            
            287   170   d_pos
            
            293   168   d_pos
            
            300   164   d_pos

            101     187   d_pos_w
            
            107     187   d_pos_w
            
            113     187   d_pos_w
            
            119     187   d_pos_w
            
            125     187   d_pos_w
            
            131     187   d_pos_w
            
            137     187   d_pos_w
            
            143     187   d_pos_w
            
            149     187   d_pos_w
            
            155     187   d_pos_w
            
            161     187   d_pos_w
            
            167     187   d_pos_w
            
            173     187   d_pos_w
            
            179     187   d_pos_w
            
            185     187   d_pos_w
            
            191     187   d_pos_w
            
            197     187   d_pos_w
            
            203     187   d_pos_w
            
            209     187   d_pos_w
            
            215     187   d_pos_w
            
            221     187   d_pos_w
            
            227     187   d_pos_w
            
            233     187   d_pos_w
            
            239     187   d_pos_w
            
            245     187   d_pos_w
            
            251     182   d_pos_w
            
            257     177   d_pos_w
            
            263     174   d_pos_w
            
            269     169   d_pos_w
            
            275     165   d_pos_w
            
            281     159   d_pos_w
            
            287     153   d_pos_w
            
            293     152   d_pos_w
            
            300     146   d_pos_w

            50       203     e
            56       203     e
            62       203     e
            68       203     e
            74       196     e
            80       196     e
            86       196     e
            92       189     e
            98       189     e
            104      189     e
            110      189     e
            116      189     e
            122      182     e
            128      182     e
            134      175     e
            140      175     e
            146      175     e
            152      175     e
            158      168     e
            164      168     e
            170      168     e
            176      161     e
            182      161     e
            188      161     e
            194      154     e
            200      154     e
            206      147     e
            212      147     e
            218      140     e
            224      140     e
            230      140     e
            236      133     e
            242      133     e
            248      126     e
            254      126     e
            260      119     e
            266      119     e
            272      119     e
            278      112     e
            284      105     e
            290      105     e
            296      98      e
            300      98      e
        };
        \addplot [mark=+,teal,ultra thick] coordinates {(150, 0) (150, 300)};

        \addplot [mark=none,black,ultra thick] coordinates {(0, 196) (320, 196)};
        \node[label,scale=1.0] at (270, 205) {\textbf{RSD} $\mathbf{=\ 196}$};
    \end{axis}
  \end{tikzpicture}
\end{figure}
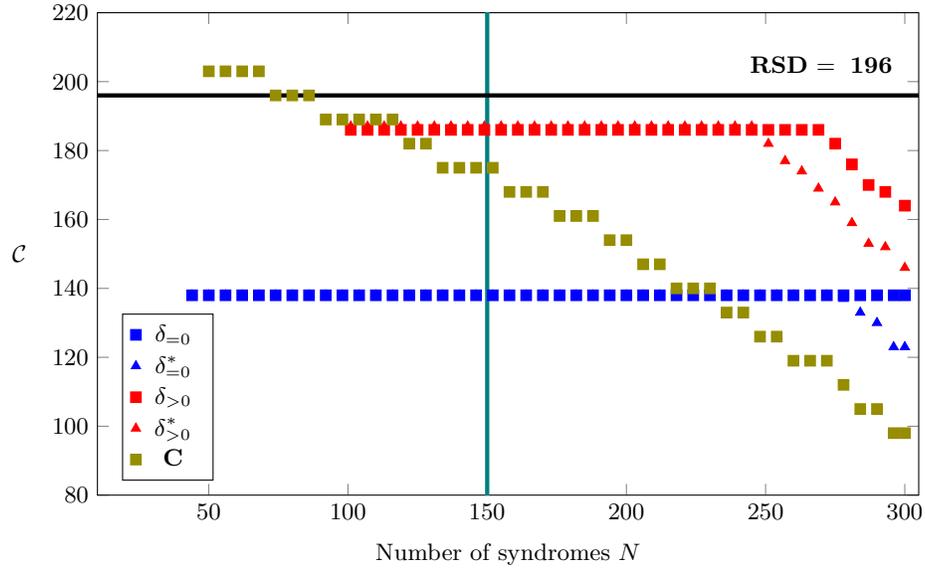

      Last but not least, thanks to our analysis of the complexity to solve \RSL with 
      different attacks, we were able to draw a graph, see Figure \ref{fig:RSL}, 
      of the complexity to solve an \RSD instance as a function of the number of given syndromes $N$. 
      
      The instance parameters are $[m,n,k,r]=[61,100,50,7]$, this is precisely the 
      instance corresponding to attacking our scheme NH-Multi-RQC-AG-128 (see Table \ref{tab:param}). 
      The complexity to solve this \RSD instance using the algebraic attack MaxMinors 
      (see Section \ref{sec:attaques_rsd}) is 196 bits; it corresponds to the horizontal black 
      thick line. 
      Starting with 44 syndromes; recall that it is the thresold for 
      the algebraic attack against \RSL (see Section \ref{sec:attaques_rsl}), 
      one sees that it beats the \RSD attack. 
      It is worth noticing that with approximately 225 syndromes, our new combinatorial 
      attack against \RSL (see Theorem \ref{thm:RSL_combi}), 
      starts to beats the algebraic attack of \cite{BB21}. 
      And finally, one notices that, with a lot of syndromes, all the aforementioned 
      \RSL attacks complexities drops down, which is quite logical.  

  \subsection{Combinatorial attack on \NHRSL}\label{sec:attaques_nhrsl}

    In this section, we adapt the combinatorial attack against \RSL, given 
    in the proof of Theorem \ref{thm:RSL_combi}, to the case of non-homogeneous 
    error, i.e. to the \NHRSL problem (see Problem \ref{problema:nhrsl}). 

    For the sake of simplicity, and since it is the case for all cryptographic 
    parameters studied in this paper, we focus only on \NHRSL instances where 
    $n_1 < n$. 

    \begin{theorem}[Combinatorial attack against \NHRSL]\label{thm:NHRSL_combi}
      There exists a combinatorial attack against an \NHRSL instance 
      with parameters $(m,n,n_1,w_1,w_2)$ whose complexity, in terms of 
      elementary operations in $\mathbb{F}_q$, is given by 
      \[
        \widetilde{\mathcal{O}}\left(q^{(w_1+w_2)(m-r)-w_2\rho}\right),  
      \]
      where $r, \rho$ are integers chosen to maximize the quantity 
      $(w_1+w_2)r+w_2\rho$ under the following constraints: 
      $N_1,N_2,r,\rho\in \mathbb{N}$,
      $N_1+N_2=N$,
      $w_1\le r$,
      $w_2\le \rho$,
      $r+\rho \le m-1$,
      $a:=\left\lfloor \frac{N_1}{w_1} \right\rfloor        \le n_1$,
      $b:=\left\lfloor \frac{N_2}{w_1+w_2} \right\rfloor    \le 2n$,
      $m(n+n_1)    \ge (n_1-b)(r+\rho)+(2n-a)r+N$.

    \end{theorem}

    \begin{proof}
      Straightforward adaptation of the attack in the proof of Theorem \ref{thm:RSL_combi}, 
      combined with the probability results given in Appendix \ref{sec:appendix_combi}. 
    \end{proof}

\section{Security and parameters of our schemes}\label{sec:params}

  \subsection{Security comparison for our schemes}

    According to Theorem \ref{thm:scheme-one},
    the security of \schemeone relies on the Decisional Ideal Rank 
    Syndrome Decoding problem (\DIRSD) and on the Decisional Ideal 
    Non-Homogeneous Rank Support Learning problem (\DNHIRSL). So far, there is no known attack to solve the decisional versions of these problems without solving the associated search instances. 
    In addition to this, there is currently no attack that takes advantage 
    of the ideal structure; thus, studying the security of \schemeone comes 
    down to evaluating the complexity of \RSD and \NHRSL.
    Unlike Multi-RQC-AG, our new scheme \schemetwo does not use ideal 
    structure. Despite the aforementioned absence of attack that exploits 
    ideal structure, it might induce a weakness in a scheme. 
    This is why Multi-RQC-AG, which does not use any structure like its name 
    suggests it, is more secure. 
    To study its complexity, according to Theorem \ref{thm:scheme-two}, one has 
    to study \RSL and \NHRSL. However, for an even better security, one could 
    use \schemetwo with homogeneous weight, making its security relying solely 
    on \RSL (see for instance the parameters sets \schemetwo-128 and \schemetwo-192 
    in Section \ref{sec:params}). 

  \subsection{Examples of parameters}

      In this section, we propose parameters for our different schemes, see Table \ref{tab:param}; 
      all parameters are chosen 
      to resist to attacks described in Section \ref{sec:security_analysis}. 
      Among the different codes that can be attacked for each of our schemes 
      (see proofs of Theorems \ref{thm:scheme-one} and \ref{thm:scheme-two}), 
      there is not a weaker one which enables us to fix all of our parameters. 
      More precisely, sometimes attacking the public key, i.e. a code $[2n,n]$, gives the lowest complexity, 
      but for another set of parameters, it will be the $[2n+n_1,n,w_1,w_2]$-code instead. 
      However, there seems to be an invariant: no matter the length $n$ of the code 
      or the dimension $m$ of the extension, it looks like the closer to GV bound the target rank $r$ is, 
      the better the combinatorial attacks are, and the worse are the algebraic attacks. 
      In other words, for a given $[m,n,k]$-code, there seems to always be a value of $r$ such that 
      all the combinatorial attacks will beat the algebraic ones. This seems to be the case 
      both for homogeneous and non-homogeneous versions of the aforementioned problems, 
      and with or without multiple syndromes. 

      \vspace{-0.5\baselineskip}

      % LABEL : tab:param
      \begin{table}[ht!]
        \centering
        {\setlength{\tabcolsep}{0.25em}
        {\renewcommand{\arraystretch}{1.4}
        {\scriptsize
        \begin{tabu}{|l|c|c|c|c|c|c|c|c|c|c|c|c|c|c|c|}
          \hline 
          \multicolumn{1}{|c|}{\multirow{2}{*}{Instance}}
          & \multirow{2}{*}{Struct.} & \multirow{2}{*}{$m$} & \multirow{2}{*}{$n'$} & \multirow{2}{*}{$n$} & \multirow{2}{*}{$n_1$} & \multirow{2}{*}{$n_2$} & \multirow{2}{*}{$k$} & \multirow{2}{*}{$\varepsilon$} & \multirow{2}{*}{$w$} & \multirow{2}{*}{$w_1$} & \multirow{2}{*}{$w_2$} 
          & \multirow{2}{*}{DFR} & \multicolumn{3}{c|}{Sizes in \textbf{KB}} \\ \cline{14-16}
          & & & & & & & & & & & & & $\mathbf{pk}$ & $\mathbf{ct}$ & Total \\ \hline \hline

          Loong-128 \cite{wang2019loong} & Random & 
            191 & 182 & 35 & 13 & 14 & 6 & 0 & 8 & 11 & 0 & 0 & 10.9 & 16.0 & \textbf{26.9} \\ \hline \hline

          Multi-RQC-AG-128 & Ideal & 
            83 & 82 & - & 5 & 38 & 2 & 74 & 7 & 11 & 0 & -138 & 0.4 & 3.9 & \textbf{4.4} \\ \hline
          NH-Multi-RQC-AG-128 & Ideal & 
            61 & 60 & - & 3 & 50 & 3 & 51 & 7 & 7 & 5 & -158 & 0.4 & 2.3 & \textbf{2.7} \\ \hline 
          Multi-RQC-AG-192 & Ideal & 
            113 & 112 & - & 4 & 60 & 2 & 98 & 8 & 13 & 0 & -215 & 0.9 & 6.8 & \textbf{7.7} \\ \hline
          NH-Multi-RQC-AG-192 & Ideal & 
            79 & 78 & - & 2 & 95 & 5 & 65 & 8 & 8 & 5 & -238 & 0.9 & 3.8 & \textbf{4.7} \\ \hline \hline 

          Multi-UR-AG-128 & Random & 
            97 & 96 & 24 & 14 & 15 & 3 & 83 & 8 & 11 & 0 & -190 & 4.1 & 6.9 & \textbf{11.0} \\ \hline 
          NH-Multi-UR-AG-128 & Random & 
            73 & 72 & 22 & 13 & 14 & 2 & 66 & 8 & 8 & 4 & -133 & 2.7 & 4.5 & \textbf{7.1} \\ \hline 
          Multi-UR-AG-192 & Random & 
            127 & 126 & 35 & 15 & 16 & 3 & 93 & 9 & 12 & 0 & -350 & 8.4 & 12.7 & \textbf{21.1} \\ \hline 
          NH-Multi-UR-AG-192 & Random & 
            97 & 96 & 30 & 14 & 14 & 3 & 77 & 9 & 9 & 4 & -214 & 5.1 & 7.5 & \textbf{12.6} \\ \hline           
        \end{tabu}
        \vspace{0.5\baselineskip}
        \caption{Parameters for our scheme \label{tab:param}}
        }}}
      \end{table}

      % \vspace{-3.5\baselineskip}

      % LABEL : tab:param_compare
      \begin{table}[ht!]
        \centering
        \vspace{2mm}
        {\setlength{\tabcolsep}{0.25em}
        {\renewcommand{\arraystretch}{1.4}
        {\scriptsize
        \begin{tabular}{|l|c|c|}
          \hline
          \multicolumn{1}{|c|}{Instance} & \hspace{0.5mm} 128 bits \hspace{0.5mm} 
            & \hspace{0.5mm} 192 bits \hspace{0.5mm} \\ \hline \hline 
          \textbf{NH-Multi-UR-AG} & \textbf{7,122}     & \textbf{12,602}    \\ \hline
          LRPC-MS   \cite{LRPC2022}              & 7,205     & 14,270    \\ \hline
          \textbf{Multi-UR-AG}    & \textbf{11,026}     & \textbf{21,075}    \\ \hline
          FrodoKEM  \cite{frodo}              & 19,336    & 31,376    \\ \hline
          Loong-128 \cite{wang2019loong}              & 26,948    & -         \\ \hline
          Loidreau \cite{pham2021etude}   & 36,300    & -         \\ \hline
          Classic McEliece  \cite{bernstein2017classic}      & 261,248   & 524,348   \\ \hline
        \end{tabular}
        \hspace{1mm}
        \begin{tabular}{|l|c|c|}
          \hline
          \multicolumn{1}{|c|}{Instance} & \hspace{0.5mm} 128 bits \hspace{0.5mm} 
            & \hspace{0.5mm} 192 bits \hspace{0.5mm} \\ \hline \hline
          \textbf{NH-Multi-RQC-AG} & \textbf{2,710}     & \textbf{4,732}     \\ \hline
          ILRPC-MS    \cite{LRPC2022}             & 2,439     & 4,851     \\ \hline
          BIKE    \cite{AABBBBDGGGMMPSTZ21}                 & 3,113     & 6,197     \\ \hline
          \textbf{Multi-RQC-AG}    & \textbf{4,378}     & \textbf{7,668}     \\ \hline
          HQC    \cite{AABBBDGPZ21a}                  & 6,730     & 13,548    \\ \hline
        \end{tabular}
        \vspace{0.5\baselineskip}
        \caption{Comparison of sizes for unstructured (random) and structured (ideal) 
        KEMs. The sizes represent the sum of the public key and the ciphertext, 
        expressed in bytes.\label{tab:param_compare}}
        }}}
      \end{table}

      \vspace{-2\baselineskip}

      Similarly to \cite{AABBBBCDGHZ20}, we use the fact 
      that $1 \in \text{Supp}(\xv,\yv)$ to set $\delta:=ww_1$ in 
      the homogeneous case and $\delta:=ww_1+w_2$ in the non-homogeneous case. 
      Recall that this quantity corresponds to the weight of the 
      error decoded by the public Augmented Gabidulin code. 
      For all our protocols (where ``NH'' denote non-homogeneous errors), both 128 and 192 bits security level are considered.
      As a comparison, we also updated 
      the parameters of the code-based KEM Loong \cite{wang2019loong}.
      Note that this scheme does not use augmented Gabidulin codes nor non-homogeneous error but 
      it does uses multiple syndromes.

      The parameters sets given in Table \ref{tab:param} come  with 
      the sizes of the associated public key $\mathbf{pk}$ and ciphertext $\mathbf{ct}$ expressed in kilo-bytes (KB).
      For \schemeone, $|\mathbf{pk}|=40 + \left\lceil \frac{n_2m}{8} \right\rceil$ and $|\mathbf{ct}|=\left\lceil \frac{2n_1n_2m}{8} \right\rceil $.
      For \schemetwo, $|\mathbf{pk}|=40 + \left\lceil \frac{nn_1m}{8} \right\rceil$ and $|\mathbf{ct}|=\left\lceil \frac{m(nn_2+n_1n_2)}{8} \right\rceil $.
      The term 40 represents the length of a seed used to generate $(\gv,\hv)$, recall that 
      the public key consists in $(\gv,\hv,\sv)$ and the ciphertext in the couple $(\uv,\vv)$. 
      Note that the size of the secret key is not relevant since it is only a seed, 
      thus it always has size 40 bytes. 
      
      To sum up, our most competitive set of parameters, in terms of sizes, 
      is NH-Multi-RQC-AG-128 which 
      uses ideal structure and non-homogeneous error; on the other side, 
      the most secure set of parameters, whose security solely depends 
      on \RSL, and which does not use ideal structure, is Multi-UR-AG-128.
      Table \ref{tab:param_compare} enables one to compare the sizes of our most competitive 
      scheme to other KEMs using ideal or random (unstructured) matrices. Note that using 
      non-homogeneous errors, our schemes are the shortest. 
      
      Last but not least, the vertical green line at $N=150$ on Figure \ref{fig:RSL} shows the number 
      of syndromes available for an adversary trying to attack a ciphertext of our scheme NH-MRQC-AG-128. 
      It is worth noticing that, even though the blue squares are below the black line (complexity of 
      the plain \RSD attack), they are still way above the security level of 128 bits, 
      and even given 150 syndromes, an attacker could not break our scheme. 
      Note that it is far away from the area where the complexities of the different \RSL 
      attacks start to drop. More generally, we picked all our parameters that way, not only 
      to resist to these attacks, but to be sure not to be targeted by any minor improvements. 
      
% \vspace*{-4mm}
      \section{Conclusion}

        In this paper, we introduce new variations on the RQC scheme, and more specifically, we
        introduce the Augmented Gabidulin codes which are very well suited to RQC. 
        These new codes, together with the multiple syndrome and the non homogeneous approaches, 
        lead to very small parameters which compare very well with other existing code-based schemes.
        In addition to this, we propose a meaningful scheme only relying on pure random 
        instances, without any ideal structure and with small parameters, around 11KBytes.
        
        We also study more deeply the security of the rank based problems used for
        our new schemes. Because of their properties, problems like \NHRSD or \RSL, 
        are probably bound to be used in many future schemes based on rank metric. 

        \section*{Ackowledgements}
        The third author would like to thank Maxime Bombar for helpful discussion. 
\appendix

% \vspace*{-3mm}
        \section{Appendices}

          \subsection{Computation of the success probability $\Pi$}
            \label{sec:appendix_Pi}\label{sec:appendix_combi}
            To compute $\Pi :=  
            \probab{S_1 \subset V,~S_2 \subset V \oplus Z}{V,Z}$ we use 
            \begin{lemma} Let $\Pi:=  
            	\probab{S_1 \subset V,~S_2 \subset V \oplus Z}{V,Z}$, where the randomness comes from the choice of a random subspace $V \subset \ff{q^m}$ and a random complementary subspace $Z$ (hence isomorphic to a subspace of $\ff{q^m}/V$). We have 
            	\begin{align*}
            	\Pi & = 
            	\probab{S_1 \subset V,~S_2/S_1 \subset (V \oplus Z \oplus S_1)/S_1}{}\\
            	&= 
            	\probab{S_1 \subset V}{V}\probab{S_2/S_1 \subset (V \oplus Z \oplus S_1)/S_1 \mid S_1 \subset V}{V,Z} \notag \\
            	& = \probab{S_1 \subset V}{V}\times \Pi_{\text{cond}}, \notag
            	\end{align*}
            	where $\Pi_{\text{cond}} := \probab{S_2/S_1 \subset (V \oplus Z \oplus S_1)/S_1 \mid S_1 \subset V}{V,Z}$.
            \end{lemma}
            \begin{proof}
            	The only non-trivial equality is the first one. For $\leq$, this is clear by taking the quotient by $S_1$. For $\geq$, let $\pi_{S_1}$ denote the quotient map $\ff{q^m} \rightarrow \ff{q^m}/S_1$. The event at the right-hand side can be seen as $S_1 \subset V,~\pi_{S_1}(S_2) \subset \pi_{S_1}(V \oplus Z \oplus S_1)$, and by considering the inverse image by $\pi_{S_1}$, this event is included in $S_1 \subset V,~\pi_{S_1}^{-1}(\pi_{S_1}(S_2)) \subset \pi_{S_1}^{-1}(\pi_{S_1}(V \oplus Z \oplus S_1))$. This gives $S_1 \subset V$ and $S_2+\ker{(\pi_{S_1})} = S_1+S_2 = S_2$. This space is included in $V \oplus Z \oplus S_1~+~\ker{(\pi_{S_1})} = V \oplus Z \oplus S_1 + S_1 = V \oplus Z \oplus S_1$, hence $S_1 \subset V$ and $S_2 \subset V \oplus Z$. \qed
            \end{proof}
            We now focus on the $\Pi_{\text{cond}}$ factor. Note that we have the decomposition
            \begin{multline*}\label{eq:disj}
            \left\lbrace S_2/S_1 \subset (V \oplus Z \oplus S_1)/S_1 \mid S_1 \subset V \right\rbrace = \left\lbrace S_2/S_1 \subset (V \oplus Z)/S_1 \right\rbrace \\
            = \coprod_{\ell=0}^{w_2} \left\lbrace \dim_{\ff{q}}(S_2/S_1 \cap V/S_1) = \ell,~
            \frac{S_2/S_1}{	S_2/S_1 \cap V/S_1} \subset 	\frac{(V \oplus Z)/S_1}{V/S_1} \right\rbrace \\ 
            = \coprod_{\ell=0}^{w_2} \left\lbrace A_{\ell} \cap B \right\rbrace,
            \end{multline*} where $A_{\ell}~:~``\dim_{\ff{q}}(S_2/S_1 \cap V/S_1) = \ell"$ and $B~:~``\frac{S_2/S_1}{	S_2/S_1 \cap V/S_1} \subset 	\frac{(V \oplus Z)/S_1}{V/S_1}"$. For $0 \leq \ell \leq w_2$, let $p_{\ell} := \probab{A_{\ell} \cap B}{}$, let $s_{\ell} := \probab{A_{\ell}}{}$ and let $t_{\ell} := \probab{B \mid A_{\ell} }{}$ so that $p_{\ell} = s_{\ell}t_{\ell}$ and $\Pi_{\text{cond}} = \sum_{\ell=0}^{w_2} p_{\ell}$. To compute $s_{\ell}$, we rely on
            \begin{lemma}[§9.3.2 p. 269, \cite{BCN89}]\label{lem:funda}
            	Let $F$ be an $\ff{q}$-linear space of dimension $n$.
            	\begin{enumerate}
            		\item If $X$ is a $j$-dimensional subspace of $F$, then there are $q^{ij}\binom{n-j}{i}_q$ $i$-dimensional subspaces $Y$ such that $X \cap Y = 0$.
            		\item If $X$ is a $j$-dimensional subspace of $F$, then there are $q^{(i-\ell)(j-\ell)}\binom{n-j}{i-\ell}_q \binom{j}{\ell}_q$ $i$-dimensional subspaces $Y$ such that $X \cap Y$ has dimension $\ell$.
            	\end{enumerate}
            \end{lemma}
            More precisely, we use Lemma \ref{lem:funda}, 2. with $F:= \ff{q^m}/S_1$, fixed $X:=S_2/S_1 \subset \ff{q^m}/S_1$ of dimension $j := w_2$ and random $Y:=V/S_1 \subset \ff{q^m}/S_1$ of dimension $i := r-w_1$. We obtain 
            \begin{equation}\label{eq:sl}
            s_{\ell} = q^{(r - w_1 - \ell )(w_2 - \ell)}\frac{\binom{m-w_1-w_2}{r-w_1-\ell}_q \times \binom{w_2}{\ell}_q}{\binom{m-w_1}{r-w_1}_q}.
            \end{equation}
            To compute $t_{\ell}$, note that conditioned on $\dim_{\ff{q}}(S_2/S_1 \cap V/S_1) = \ell$ the probability that $\frac{S_2/S_1}{	S_2/S_1 \cap V/S_1} \subset 	\frac{(V \oplus Z)/S_1}{V/S_1}$ is the probability that a random subspace of dimension $\rho$ contains a fixed subspace of dimension $w_2 - \ell$ in the ambient space $\frac{\ff{q^m}/S_1}{V/S_1} \simeq \ff{q^m}/V$. From there we obtain $t_{\ell} = \frac{\binom{\rho}{w_2 - \ell}_q}{\binom{m-r}{w_2 - \ell}_q}$, and finally by combining this with Equation \eqref{eq:sl}:
            \begin{eqnarray*}
            	p_{\ell} = q^{(r - w_1 - \ell )(w_2 - \ell)}\frac{\binom{m-w_1-w_2}{r-w_1-\ell}_q \times \binom{w_2}{\ell}_q}{\binom{m-w_1}{r-w_1}_q} \times \frac{\binom{\rho}{w_2 - \ell}_q}{\binom{m-r}{w_2 - \ell}_q}.
            \end{eqnarray*}

            \subsection{Finishing the proof of Theorem \ref{theo:complex_combi}}\label{sec:K}

             Obviously $p_0 < \Pi_{cond}$ and one can also easily show that $\Pi_{cond} = \Theta{(p_0)}$. By including the $q^{-m}$ factor from \cite{AGHT18}, the number of $\ff{q}$-operations in the attack is 
              $$\mathcal{K} = \mathcal{O}\left(L \times \Pi^{-1} \times q^{-m} \right) = \widetilde{\mathcal{O}}\left( \probab{C}{}^{-1}p_0^{-1} q^{-m} \right),$$
              where $\probab{C}{} := \probab{S_1 \subset V}{V}$ and where $L$ is the polynomial factor coming from the linear algebra step whose exact formula is not relevant for the discussion. Using the classical $\textstyle\binom{a}{b}_q = \Theta(q^{b(a-b)})$ when $\max{(a,b)} \rightarrow +\infty$ together with
              \begin{align*}
              p_{0} & = q^{ (r-w_1)w_2}\frac{\binom{m-w_1-w_2}{r-w_1}_q}{\binom{m-w_1}{r-w_1}_q}\frac{\binom{\rho}{w_2}_q}{\binom{m-r}{w_2}_q},
              \end{align*}
              we obtain
              $p_{0} = \Theta{\left( q^{ (r-w_1)w_2} \times q^{-(r-w_1)w_2} \times q^{-w_2(m-r-\rho)}\right)} = \Theta{(q^{-w_2(m-r-\rho)})},$
              and similarly $\probab{C}{} = \Theta{\left( q^{-w_1(m-r)}\right)}$. Therefore $\mathcal{K} =  \widetilde{\mathcal{O}}{\left(q^{(w_1+w_2)(m-r) - w_2\rho - m}\right)}$, which is the statement of Theorem \ref{theo:complex_combi}.

        \subsection{Proofs for the MaxMinors attack on \NHRSD}

          \subsubsection*{Proof of Proposition \ref{prop:MM_easy}.}\label{sec:proof_prop1}For item 1., let $J \subset \{1..n+n_1-1\},~\# J = w_1 + w_2$ such that $P_J \in \mathcal P_{\text{lost}}$. By definition of $\mathcal P_{\text{lost}}$ the set $J+n+1$ has intersection $\leq w_2 - 2$ with $\{n+2..n+n_1\}$, hence any subset $T = T^-\cup T^+,~T^-\subset \{1..n+1\},~T^+\subset (J+n+1)$ satisfies $\#(T\cap\lbrace n+1..n+n_1\rbrace) \leq w_2-1$ since $T^{-}$ might also contain $n+1$. This means that the corresponding minor variable $c_T$ belongs to $\zeta$ and can be set to zero in $P_J$. Using the shape depicted in Equation \eqref{eq:P_J}, this implies that the whole $P_J$ equation becomes zero. For item 2. , recall that the leading term of $P_J \in \mathcal P_{\text{indep}}$ is $c_{J+n+1}$. Moreover we have $\#(J \cap\lbrace 1..n_1-1\rbrace) = \#(J + n + 1\cap\lbrace n+2..n+n_1\rbrace) \geq w_2$, which means $c_{J+n+1} \notin \zeta$. In particular, all the equations from $\mathcal P_{\text{indep}}$ keep the same leading terms after fixing the $M$ minor variables to zero and therefore they remain linearly independent. The last statement on the number of variables in obvious.

          \subsubsection*{Lemmata to prove Proposition \ref{prop:MM}.}\label{sec:proof_prop2}

          \paragraph{Proof of Lemma \ref{lem:P_resta}.} Using Equation \eqref{eq:P_J_tilde}, one has that the equations in 
          $\mathcal{P}_{\text{rest},A}$ all have their monomials in 
          $  \mu_A := \big\lbrace c_T,~T \subset \{1..2n+n_1\},~\#T = w_1+w_2,~n+1 \in T, \\ 
            T \cap \{n+2..n+n_1\} = A \big\rbrace,$
          and this set has size $\textstyle \binom{2n}{w_1}$. Finally, note that $\mu_A$ and $\mu_A'$ are disjoint when $A \neq A'$, which concludes the proof.

          \paragraph{Proof of Lemma \ref{lem:dim_P_resta}, under assumptions.} Using Equation \eqref{eq:P_J_tilde}, it is readily verified that the set of leading terms of all equations in $\mathcal P_{\text{rest},A}$ is 
          $$\tau_A := \left\lbrace c_{\lbrace n+1 \rbrace \cup A \cup U},~U \subset \{(n+n_1 +2)..(2n+n_1)\},~\# U = w_1 \right\rbrace,$$
          and for instance note that the equation $P_{J_U}$ with $J_U + n+1 = A \cup \{n+n_1+1\} \cup U$ has leading term $c_{\lbrace n+1 \rbrace \cup A \cup U} \in \tau_A$. This already shows that $\dim_{\ff{q^m}}\left \langle \mathcal P_{\text{rest},A} \right\rangle \geq \# \tau_A  = \textstyle \binom{n-1}{w_1}$. For the converse inequality, we need to rely on some assumption on the randomness of the entries of the $P_J$'s in $\ff{q^m}$ to argue that we cannot construct an element in $\left \langle \mathcal P_{\text{rest},A} \right\rangle$ whose leading term does not belong to $\tau_A$ with very high probability. First, note that the variables from $P_J \in \mathcal{P}_{\text{rest},A}$ with $J+n+1 = A \cup V_J$ where $V_J = \left\lbrace v^{(J)}_{1} < \dots < v^{(J)}_{w_1+1}\right\rbrace$ which belong to $\tau_A$ are the $c_{\lbrace n+1 \rbrace \cup A \cup V_J \setminus \lbrace v^{(J)}_{j} \rbrace}$ for $1 \leq j \leq w_1 + 1$. To kill the leading term of $P_J$, one would then consider an equation with the same leading term, namely a $P_J'$ with $J' \neq J,~J'+n+1 = A \cup V_{J'}$ and such that $V_{J \setminus \lbrace v^{(J)}_{1} \rbrace} = V_{J'
          	\setminus \lbrace v^{(J')}_{1} \rbrace} = B$ for some $B$. In this case, one can check that the only monomial from $\tau_A$ present in both $P_J$ and $P_{J'}$ is $c_{\lbrace n+1 \rbrace \cup A \cup B}$, so that $P_J + \lambda_{J'}P_{J'}$ contains at least $2w_1$ monomials from $\tau_A$. Similarly, by using a third $J''$ one could kill at most one extra monomial in $P_J$ and in the worst case one in $P_{J'}$ as well. This means that a linear combination of the form $P_J + \lambda_{J'}P_{J'} + \lambda_{J''}P_{J''}$ contains at least $2(w_1 - 1) + (w_1 + 1 - 2) =3(w_1 - 1)$ monomials from $\tau_A$, and the lower bound is reached if and only if those monomials in $P_J$ and $P_J'$ are killed at the same time by $\lambda_{J''}P_{J''}$. This is extremely unlikely if the coefficients of the MaxMinors equations are random elements in $\ff{q^m}$, so that we assume instead that $P_J + \lambda_{J'}P_{J'} + \lambda_{J''}P_{J''}$ contains at least $(w_1 - 1) + w_1 + (w_1 + 1 - 1) = 3w_1 - 1$ monomials in $\tau_A$. Relying on the same type of assumption, one can proceed by induction on the numbers of terms to show that a non-zero linear combination in $\left \langle \mathcal P_{\text{rest},A} \right\rangle$ always has a monomial in $\tau_A$.

% \vspace*{-3mm}
        \bibliographystyle{splncs04}
        \bibliography{multi_rqc}
        \end{document}